\relax
\documentclass[]{article}
\usepackage{times}  
\usepackage{helvet} 
\usepackage{courier}  
\usepackage[hyphens]{url}  
\usepackage{graphicx} 
\usepackage{graphicx}  
\usepackage[left=1in,right=1in,top=1in,bottom=1in]{geometry}

\usepackage{enumerate}

\usepackage{multirow,bigdelim}
\usepackage{mathtools}
\usepackage{amsmath}
\usepackage{amsthm}
\usepackage{bm}
\usepackage{bbm}
\usepackage{amssymb}
\usepackage[]{algorithm2e}
\usepackage{graphicx}
\usepackage{amsmath}
\usepackage[cmyk]{xcolor}
\selectcolormodel{cmyk}
\usepackage{comment}
\usepackage{thmtools}
\usepackage{thm-restate}
\usepackage{cleveref}
\usepackage[]{algorithm2e}
\usepackage{pgfplots}

\newcommand{\Xcomment}[1]{{}}

\newtheorem{theorem}{Theorem}[section]

\newtheorem{lemma}[theorem]{Lemma}
\newtheorem{definition}[theorem]{Definition}
\newtheorem{claim}[theorem]{Claim}

\newtheorem{conjecture}[theorem]{Conjecture}

\theoremstyle{remark}

\newtheorem{example}[theorem]{Example}

\theoremstyle{plain}

\newcommand{\noaccents}[1]{#1}
\newcommand{\newagentvar}[3][\noaccents]{%
\expandafter\newcommand\expandafter{\csname #2\endcsname}{#1{#3}}%
\expandafter\newcommand\expandafter{\csname #2s\endcsname}{#1{\boldsymbol{#3}}}%
\expandafter\newcommand\expandafter{\csname #2smi\endcsname}[1][i]{#1{\boldsymbol{#3}}_{-##1}}%
\expandafter\newcommand\expandafter{\csname #2i\endcsname}[1][i]{#1{#3}_{##1}}%
\expandafter\newcommand\expandafter{\csname #2ith\endcsname}[1][i]{#1{#3}_{(##1)}}%
}

\newcommand{\newvecagentvar}[3][\noaccents]{%
\expandafter\newcommand\expandafter{\csname #2\endcsname}{#1{\boldsymbol{#3}}}%
\expandafter\newcommand\expandafter{\csname #2s\endcsname}{#1{\boldsymbol{#3}}}%
\expandafter\newcommand\expandafter{\csname #2smi\endcsname}[1][i]{#1{\boldsymbol{#3}}_{-##1}}%
\expandafter\newcommand\expandafter{\csname #2i\endcsname}[1][i]{#1{\boldsymbol{#3}}_{##1}}%
\expandafter\newcommand\expandafter{\csname #2ith\endcsname}[1][i]{#1{#3}_{(##1)}}%
}

\newagentvar{val}{v}
\newagentvar{bid}{b}
\newagentvar{dist}{F}
\newagentvar{alloc}{x}
\newagentvar{util}{u}
\newagentvar{pay}{p}
\newagentvar{ability}{a}
\newagentvar{effort}{e}

\usepackage{setspace}

\title{VCG Under False-name Attacks: a Bayesian Approach\footnote{A preliminary and partial version of this paper is to appear in Gafni, Y., Lavi, R. and Tennenholtz M. 2020. VCG Under Sybil (False-name) Attacks - a Bayesian Analysis. In Proceedings of the 34th National Conference on Artificial Intelligence (AAAI-20)}} 

\author{\Large{Yotam Gafni$^1$, Ron Lavi$^{1,2}$ and Moshe Tennenholtz$^1$} \\ 
$^1$ Technion - Israel Institute of Technology \\
$^2$ University of Bath, UK \\
\{yotam.gafni@campus, ronlavi@ie, moshet@ie\}.technion.ac.il
}
\date{}

 \pgfplotsset{compat=1.13}
 
\begin{document}
\onehalfspacing
\maketitle

\begin{abstract}
VCG is a classical mechanism that is often applied to combinatorial auction settings. However, while the standard single-item Vickrey auction is false-name-proof, a major failure of multi-item VCG is its vulnerability to false-name attacks. This occurs already in the natural bare minimum model in which there are two identical items and bidders are single-minded. Previous solutions to this challenge focused on developing alternative mechanisms that compromise social welfare. We re-visit the VCG vulnerability and consider the bidder behavior in Bayesian settings. In service of that we introduce a novel notion, termed the {\em granularity threshold}, that characterizes VCG Bayesian resilience to false-name attacks depending on the bidder type distribution, with emphasis on the granularity of items demand. For general distributions, we characterize the existence of truthful BNE for two bidders, and show that for three or more bidders, truthful BNE does not necessarily exist. We then give a computational method for verifying the existence of truthful BNE for the class of beta distributions, and provide empirical results. The results demonstrate the importance of an intuitive form of false-name attack we call ``the split attack'' to the truthful BNE analysis of the uniform distribution. 
\end{abstract}

\section{Introduction}
In recent years, the scale of web auctions has grown significantly ---  whether for ads, flight tickets or cloud computing resources ---  among many other commodities. Contrary to the settings that apply to government leases or items sold in an auction house, in such web auctions it is much easier to create false identities and submit multiple bids under false names. This expands the action space of an auction participant and may turn mechanisms that are provably strategy-proof, assuming only a single bid per agent, to not be strategy-proof, which is the case of the most classic combinatorial auction mechanism, VCG. VCG is truthful and efficient assuming each bidder submits one bid under her name \cite{agt}. However, assuming false-name attacks (a.k.a. sybil attacks) are possible, that is no longer the case. To demonstrate this, Figure~\ref{fig:BeneficialFalseName} adapts an example of \cite{yokoo2002}.
\begin{figure}
    \hspace{2cm} \begin{tabular}{c c|c|c|c|c}
           & \textbf{Bidder} & \textbf{True type} & \textbf{Bids} & \textbf{Wins?} & \textbf{Payment} \tabularnewline
         \hline
            \hspace{-4cm} \ldelim\{{2}{3mm}[Truth] & Bidder 1 & (2,0.4) & (2,0.4) &  &  \tabularnewline
         \hline
         & Bidder 2 & (2,0.5) & (2,0.5) & X & 0.8 \tabularnewline
    \hline
    \hline

         \hline
          \hspace{-4cm} \ldelim\{{3}{3mm}[Attack] & Bidder 1 & (2,0.4) & (2,0.4) & & \\
        \hline
         &  Bidder 2 & (2,0.5) & (1,0.5) & X & 0.3 \\
         \hline
         & &  & (1,0.5) & X & 0.3 \\
    \end{tabular}
    \caption{A beneficial false-name attack}
    \label{fig:BeneficialFalseName}
    \end{figure}
    
    \begin{figure}
\hspace{2cm}
    \begin{tabular}{c|c|c|c|c|c}
         & \textbf{Bidder} & \textbf{True type} & \textbf{Bids} & \textbf{Wins?} & \textbf{Payment}\\
         \hline
         \hspace{-4cm} \ldelim\{{2}{3mm}[Truth] & Bidder 1 & (1,0.4) & (1,0.4) & & \\
        \hline
         & Bidder 2 & (2,0.5) & (2,0.5) & X & 0.4 \\
    \hline
    \hline
         \hline
         \hspace{-4cm} \ldelim\{{3}{3mm}[Attack] & Bidder 1 & (2,0.4) & (2,0.4) & & \\
         \hline
         & Bidder 2 & (2,0.5) & (1,0.5) & X & 0.4 \\
         \hline
         & &  & (1,0.5) & X & 0.4 \\
    \end{tabular}
    \caption{A detrimental false-name attack}
    \label{fig:BayesianMotivation}
\end{figure}
    In this example, two bidders bid for two auctioned items. In the figure, bidder types are denoted by $(2,\theta)$ with $\theta = 0.4, 0.5$ respectively, which means that
    each bidder requires both items (has utility zero from obtaining only one of the items), and is willing to pay $\theta$ per item, so overall values the two items as $2\theta$. When bidder $2$ bids truthfully, it wins the two items and pays $0.8$. But when bidder $2$ attacks by creating two identities, each bidding $(1,0.5)$, it can be directly calculated that each of these bids wins one item and pays $0.3$, resulting in a total payment of $0.6$ for bidder $2$. By performing this false-name attack, bidder $2$ was able to reduce its payment.

A main effort to address this issue is by Yokoo et al.~who develop various mechanisms that are {\em false name attack proof}~\cite{yokoo2001,yokoo2003}. The main disadvantage of these mechanisms is their low social welfare, which intuitively can be explained by the fact that they incorporate the false-name attacks to be built into the mechanism. In \cite{Iwasaki} the main result shows that under reasonable conditions the worst case social welfare of any false-name proof  combinatorial auction with $m$ items is at most a fraction of $\frac{2}{m+1}$ of the optimal social welfare. 

As for VCG, in \cite{vetta2014} the authors show that under some assumptions on the bidders' valuations, and in a setting with complete information, if a  pure Nash equilibrium exists when false-name bids are considered, VCG still has a reasonable welfare guarantee. Naturally, Nash equilibria might not exist.
Nevertheless, the above may hint that VCG  may actually be a good false-name proof mechanism in situations in which a pure equilibrium exists. 

Notice that when the behavior of the other bidders is uncertain, in some scenarios a false-name attack may benefit the attacker, while on others the attack can harm her. 
So while truthful bidding is certainly no longer a dominant strategy, it is worth asking under what circumstances it might be a {
\em Bayesian Nash Equilibrium}.
That is, when the auction participants only have probabilistic information regarding other bidders' valuations, when is it the best strategy for a bidder to bid truthfully 
(in a single bid) in the VCG mechanism, assuming all others are truthful as well? Notice that social welfare is optimized in such an equilibrium. If a truthful equilibrium exists, it is quite natural to assume that it will be realized/selected, for at least two reasons. First, if it can be shown that truthful bidding is a Bayesian Nash equilibrium for a {\em class} of distributions, and thus selecting the truthful equilibrium does not require knowing the exact underlying distribution of types in this class; a (common) knowledge that the underlying distribution belongs to this class is sufficient. The truthful equilibrium is thus a natural Schelling point.
On the other hand, the other non-truthful equilibria are more likely to rely on the exact details of the distribution, and may be harder to come by.

In this work, we introduce ---  for the first time, to the best of our knowledge ---  a Bayesian equilibrium analysis of VCG under false-name bids. In service of that we use a bare-minimum model: an auction of two identical items and single-minded bidders. 
In this model, each bidder is interested in either a single item (in this case a second item does not add extra value) or in the pair of items (in this case a single item does not have any positive value). The probabilistic setting is given by a per-item valuation distribution, and a probability $q$ for a bidder to have a single-item demand. 

Hence, our model captures a step forward from the standard Vickrey auction setting to the general multi-item case. As we will see, the analysis of such a model already brings out intricate techniques and conclusions, which are essential to addressing the problem in its full generality. Specifically, while the above model is indeed a significant simplification, it allows us to make several contributions which we believe are meaningful:
\begin{itemize}
\item We show several impossibilities even for this simple model.
\item We believe that these impossibilities should be viewed as worst-case exceptions and that under natural restrictions on the assumed model the general concept of truthful bidding as an equilibrium in VCG auctions with false-name bids is useful. Our model allows us to demonstrate this point by showing several positive results. It thus advances the state of art and puts forth the long term goal of determining the exact possibility-impossibility border. 
\item Even this simple model gives rise to several high-level conceptual conclusions which we discuss below, regarding the effect of the number of bidders on truthful bidding, the structure of beneficial false-name attacks, etc.
\end{itemize}

\subsection{Overview of Main Results}

To demonstrate the promising nature of the Bayesian approach, we revisit the example given before. In Figure~\ref{fig:BayesianMotivation}, bidder 2 with type $(2,0.5)$ now faces a different adversary with the per-item value $0.4$, but with a single item demand rather than a double item demand. As in the first example, bidder 2 wins the two items both when bidding truthfully and attacking, but now pays $0.4$ more when attacking. 

This hints at the potential benefit of a Bayesian approach: The bidder might not submit false-name bids as her potential losses could be higher than her potential gains. Moreover, the only difference between the two examples is the number of items demanded by the adversary bidder. When the adversary demands two items, the false bidders are able to ``free-ride'' off each other to reduce their payments. On the other hand, when the adversary demands one item, the false bidders take ``double responsibility'' for the losing bid. In this paper we formalize using a notion of a ``granularity threshold'' the intuition that facing more granular competing bids, i.e., a larger probability $q$ of facing bidders that demand one item, plays an important part in deterring false-name attacks. 

Our main result states that for two bidders ($n=2$) and two items ($m=2$) there is a truthful Bayesian-Nash equilibrium (BNE) whenever the probability $q$ of a bidder to demand one item rather than two is at least $\frac{2}{3}$, for {\em any} per-item valuation distribution (Section~\ref{sec:global_granularity_n=2}). 
We then show that for any larger number of bidders $(n>2)$, such a global granularity threshold can no longer be derived and valuation distributions do matter (Section \ref{sec:impossibility_result}). Thus, in sharp contrast to the case of two bidders, when we have at least three bidders truthful BNE does not exist for {\em all} per-item value distributions.

Should we interpret this impossibility result as saying that truthful BNE only exists in very limited cases? We believe not, and to further study this issue we restrict our attention to a large class of per-item value distributions commonly termed beta distributions with support $[0,1]$. This family includes many natural distributions and is used widely 
in statistics and economic theory \cite{BetaDist,BetaDist2}. It includes all distributions $F(x)=x^\alpha$ for any integer $\alpha \geq 1$ (thus also the uniform distribution by taking $\alpha=1$), and with large $\alpha=\beta$ parameter values it resembles a Normal distribution. 
Interestingly, it admits useful connections to computer algebra techniques, which allows us to provably present good granularity thresholds, i.e. low $q$ values, in tested cases (Section \ref{sec:beta_distributions}).
While it remains an open question to determine granularity thresholds for all $\alpha,\beta$ parameters, one high-level conclusion from our analysis is that the granularity threshold is significantly smaller than 1 for several such parameters.
The disadvantage of the computer algebra approach is that it is able to give results for concrete polynomials, i.e., for some fixed parameter values (including the number of bidders). In Section~\ref{sec:split_attack} we show results for a general number of bidders $n$ for the uniform distribution. The analysis reveals some interesting properties of a specific form of attack which we term \emph{a split attack}.

Our analysis yields conceptual conclusions which we discuss in Section~\ref{sec:Discussion}. For example, it contradicts an initial intuition that may suggest that false-name attacks are beneficial when there are fewer bidders and in particular in the case of two bidders. This intuition follows the observation that the benefit of a false-name attack comes from reducing payments by free-riding on false identities. Such tailored free riding seems less plausible when there are many competing bidders. However, our empirical results in Section~\ref{sec:beta_distributions} clearly demonstrate that this intuition is false and that the granularity threshold increases with the number of bidders. Our theoretical results also contradict this intuition by showing a general possibility for $n=2$ and a general impossibility for $n \geq 3$. 


\section{Model and Preliminaries}
We study a multi-item auction with two identical items and $n$ single-minded bidders. The type of a bidder $i$ is denoted by
$\hat{\theta_i} = (g_i, \theta_i)$ where $g_i \in \{1,2\}$ is the number of items desired by the bidder, and $\theta_i$ is the per-item value of the bidder. We assume that the per-item values are normalized to be in $[0,1]$. The utility of a bidder is quasi-linear, i.e. if the bidder receives $j$ items and pays a price $p_i$, her utility is 
$$ 
\begin{cases}
-p_i & j < g_i \\
g_i \cdot \theta_i - p_i & otherwise. \\
\end{cases}
$$


In this paper, we focus on the analysis of the VCG mechanism in a setting where each bidder may submit multiple single-minded bids in an anonymous fashion so that VCG treats each bid as if it comes from a separate bidder. Formally, VCG is the mechanism that given bids $B = (\hat{\theta}_1, ..., \hat{\theta}_N)$ (where $N \geq n$) allocates the items to maximize the social welfare:

    $$ SW(B) = \max\limits_{\substack{I \subseteq \{1, ..., N\} \\ \sum_{i\in I}g_i \leq 2}} \sum_{i\in I} g_i\cdot \theta_i. $$

Let $W(B)$ be the index set $I$ of bidders in a maximal social welfare allocation, and let $L(B)$ be its complement. Every bid $b$ prompts the following payment:

$$
\begin{cases}
SW(B\setminus \hat{\theta}_b) - (SW(B) - g_b \cdot \theta_b) & b \in W(B) \\
0 & b \in L(B). \\
\end{cases}
$$

That is, the bidder pays the difference between the others' utilities in her absence and the others' utilities when it participates. Every real bidder $i \in \{1,...,n\}$ receives her value from the union of items won by her submitted bids, and pays the respective sum of prices.

Throughout our analysis, given the symmetric ex-ante situation between the bidders, we fix one bidder, bidder $n$, and analyze her utility. The bidder may declare her true type or submit any amount of false-name bids. For the analysis of the truthful Bayesian-Nash equilibrium (BNE) we assume all other bidders reveal their true types. We term the case where the bidder submits false-name bids --  an ``attack''. In our analysis, we refer to the other bidders as the \emph{``adversary''} bidders. We introduce the notation $\tilde{n} = n-1$, the number of adversary bidders. 
We refer to bidder $i$ with $g_i=1$ as a ``1-type'', and denote her value by $v_i$. Similarly, we refer to a bidder $i$ with $g_i=2$ as a ``2-type'', and denote her per-item value $w_i$. We sometimes use two separate indices, for 1-types and for 2-types. Formally, let $k = |\{1 \leq i \leq \tilde{n}, g_i = 1\}|$, the number of 1-type bids. We denote the 1-type adversary bids by $v_1,...,v_k$ and the 2-type values by $w_1,...,w_{\tilde{n}-k}$, both by decreasing order. If we refer to a bid value which doesn't exist, e.g. we refer to $v_1$ when there are no adversary 1-type bids, we assume by default its value is 0. 

We denote the true type of bidder $n$ as $\hat{\theta}$. 
Bidder $n$ could bid $(\hat{\theta'_1},...,\hat{\theta'_m})$ under $m$ identities.  
We denote bidder $n$'s utility (with own bids, true type, as parameters and the adversary bids as the variables), 
$$u_{\hat{\theta'_1},...,\hat{\theta'_m}}^{\hat{\theta}}(\hat{\theta_1}, \ldots, \hat{\theta}_{\tilde{n}}).$$

\subsection{A Bayesian Analysis} 




We now move to formalize the intuition of the examples given in the introduction. We do so using a standard Bayesian analysis approach, and assume that the true types of the bidders are independently and identically drawn from some distribution. Most generally, this distribution can be parameterized by $0\leq q \leq 1$, which is the  probability of a bidder $i$ to have $g_i=1$, and by two per-item value distributions $F_1, F_2$ with normalized support in $[0,1]$ which give the per-item value distribution of the bidder given her item demand $g_i$. In this paper we focus on a special case where $F_1=F_2\stackrel{def}{=}F$ (see further discussion of the general case in Section~\ref{sec:Discussion}). Without loss of generality we assume that $F(\theta) < 1$ for all $\theta<1$.
As is standard, we assume each bidder knows her own realized type and the distribution parameters $q,F$. In this setting, we examine whether truthful bidding is a Bayesian Nash equilibrium (BNE). This requires that for every realized true type in the support of the distribution, the expected utility of a truthful bidder is not less than her expected utility resulting from any attack, where the expectation is taken over the types of the other bidders (adversary types) and assuming all adversaries submit truthful bids.
%
Formally, consider first a bidder that knows the realized number of adversary 1-type bidders $k$. Given her true type $\hat{\theta}$, her bids $S$ (either the true type or some form of attack), the per-item distribution $F$, and the number of bidders $n$ (recall that $\tilde{n}\stackrel{def}{=}n-1$), her expected utility is:
\[
\label{eq:expectedUtilK}
\begin{split}
& E^{k,F,n}_{\hat{\theta},S} \stackrel{def}{=} E_{\substack{v_1,..., v_k \sim F \\ w_1,...,w_{\tilde{n}-k} \sim F}}[u^{\hat{\theta}}_{S}((1,v_1),...,(1,v_k),(2,w_1),...,(2,w_{\tilde{n}-k}))] 
\end{split}
\]
In words, the expectation over $k$ 1-type adversaries and $\tilde{n}-k$ 2-type adversaries with per-item values sampled from $F$. Further define, given the parameter $q$ (rather than some fixed $k$ number of 1-type bidders),

$$E^{q,F,n}_{\hat{\theta},S} \stackrel{def}{=} E_{k \sim Binomial(\tilde{n},q)} [E^{k,F,n}_{\hat{\theta},S} | k] = \sum_{k=0}^n \binom{\tilde{n}}{k} q^k (1-q)^{\tilde{n}-k} E^{k,F,n}_{\hat{\theta},S},$$
namely the number $k$ of 1-type adversary bidders is drawn from the binomial distribution over $\tilde{n}$ with parameter $q$ (since bidder types are i.i.d.).

\begin{definition}For a set of parameters $q,F,n$, we say that truthful bidding is a Bayesian Nash Equilibrium (BNE) if for any bidder's type $\hat{\theta}$ in the support of $F$ and attack $\hat{\theta}_1,\ldots,\hat{\theta}_m$ (where each attack bid is in the support of $F$ as well),
\[
\begin{split}
E^{q, F,n}_{\hat{\theta},truth}[u] \geq E^{q,F,n}_{\hat{\theta},(\hat{\theta}_1',\ldots,\hat{\theta}_m')}[u].
\end{split}
\]
\end{definition}
We say that an attack is ``beneficial'' for some given adversary bids if there exists a true type such that the attack increases a bidder's utility over her truthful bid. If we do not specify concrete adversary bids, an attack is ``beneficial'' if it increases the bidder's expected utility.

Our first result shows that it is possible to reduce the analysis to considering only attacks of the type $(1,x),(1,y)$. For this type of attacks, that constitute two 1-type bids, we use the notation $attack(x,y)$. It also shows that we can further restrict parameters if bidder $n$'s true type is a 1-type. In this case it is enough to consider $\theta = 1$:  

\begin{restatable}[]{theorem}{attackformat}
\label{thm:attack_format}
For a set of parameters $q,F,n$, truthful bidding is a Bayesian Nash Equilibrium (BNE) if and only if it holds that 
\[
\begin{split}
E^{q, F,n}_{\hat{\theta},truth}[u] \geq E^{q,F,n}_{\hat{\theta},attack(x,y)}[u].
\end{split}
\]


\noindent
for any $x,y,\hat{\theta}=(g,\theta)$ such that $x,y, \theta \in supp F$ and, if $g=1$, then $y \leq x \leq \theta = 1$.

\end{restatable}

\noindent
The proof is given in Appendix A. 


\subsection{A Criterion for Bayesian Resilience}

As noted before, and demonstrated by the examples in Figures~\ref{fig:BeneficialFalseName}~and~\ref{fig:BayesianMotivation}, false-name attacks perform poorly when facing adversaries with 1-item demand.
In the extreme case, where adversaries are guaranteed to be 1-type, false-name attacks are not beneficial:

\begin{restatable}[]{lemma}{qanalysis}
\label{lem:q=1_analysis}
When all adversaries are 1-type, an attack is never beneficial, regardless of $F$ and $n$. Thus, $q=1$ always induces a truthful BNE (for every $F$ and $n$). 
\end{restatable}

\noindent
The proof of the lemma is given in Appendix B.

A possible goal in our context is to characterize the class of all settings (i.e., common priors) $q,F,n$ that admit a truthful BNE. The above lemma motivates us to try to consider ``large enough'' $q$ values, as captured by the following definition:

\begin{definition}
The ``granularity threshold'', 
$$q^*_{n,g,\theta,x,y,F} = min \left\{ \bar{q} \text{ s.t. } \forall q \in (\bar{q}, 1], E^{q, F,n}_{\hat{\theta},truth}[u] \geq E^{q,F,n}_{\hat{\theta},attack(x,y)}[u] \right\}.$$
In words, the minimal $\bar{q}$ such that for all $q > \bar{q}$ an attack $x,y$ is not beneficial for a true type $(g,\theta)$ given that there are overall $n$ bidders and distribution $F$.

When we omit some of the parameters in the subscript of $q^*$, the omitted parameters should appear in the right side with a for-all qualifier. Most importantly,

$$q^*_{n,F} = min \left\{ \bar{q} \text{ s.t. } \forall q \in (\bar{q}, 1], g,\theta,x,y,\ E^{q, F,n}_{\hat{\theta},truth}[u] \geq E^{q,F,n}_{\hat{\theta},attack(x,y)}[u] \right\}.$$








\end{definition}

The granularity threshold is a possible way to measure ``Bayesian resilience'', since as $q^*$ decreases, there are more settings in which truthful bidding is a BNE in VCG, in particular, all settings $q,F,n$ for which $q \geq q^*_{n,F}$. Roughly speaking, it lower bounds the measure of the class of common-prior distributions that guarantee truthful bidding, in the following sense. Fix some per-item value distribution F. Then, if the common-prior distribution over types, parameterized by $q,F$, has $q \geq q^*_{n,F}$ then truthful bidding is a BNE. Similarly, taking $q \geq q^*_{n}$, truthful bidding is a BNE for every
common-prior distribution over types parameterized by $q,F$ (for every F). Thus, a smaller value of $q^*$ means that a larger class of common-prior distributions yield truthful bidding as a BNE of the VCG auction with false-name bids.

\section{A Full Characterization of the Granularity Threshold for $n=2$}
\label{sec:global_granularity_n=2}
In this section, we characterize $q^*_{n=2}$, i.e., the granularity threshold assuming 2 bidders 
that can be of any true type, use any form of attack, and the distribution over bidder types can be any distribution $F$. To analyze this, recall that we assume that bidder~1 truthfully reveals her type which is drawn from the given distribution and under this assumption we examine which bids maximize the expected utility of bidder~2. We do so by considering two cases - one where we restrict bidder $2$ to be a 1-type ($q^*_{n=2,g=1})$, and one where we restrict bidder $2$ to be a 2-type ($q^*_{n=2,g=2})$. Overall we show:

\begin{theorem}
\label{thm:global_granularity_n=2}
$$q^*_{n=2} = \max_{g'\in \{1,2\}} q^*_{n=2,g=g'} = \max \{\frac{1}{2},\frac{2}{3}\} = \frac{2}{3}.$$
\end{theorem}

The remainder of this section proves the theorem.

\begin{lemma}$q^*_{n=2,g=1}=\frac{1}{2}$.
\label{lem:one_type_q*}
\end{lemma}
\begin{proof}
Put explicitly, the lemma condition requires that for every $q > \frac{1}{2}$, every true per-item value $\theta$ and attack parameters $x,y$, we have  
\begin{equation}
    \label{eq:BNEcondition_n=2}
    E^{q,F,n=2}_{\hat{\theta}=(1,\theta),truth}[u] - E^{q,F,n=2}_{\hat{\theta}=(1,\theta),attack(x,y)}[u] \geq 0.
    \end{equation}
Theorem~\ref{thm:attack_format} allows us to assume $0\leq y\leq x \leq \theta = 1$. For $n=2$ we have $\tilde{n}=1$ and so we have one adversary, either of type 1 or 2. Since some parameters are constant throughout the proof, we simplify our previous notations and write
$u^{truth}_{(g_{adv},\theta_{adv})} = u^{\hat{\theta}=(1,1)}_{\theta'_1=(1,1)}((g_{adv},\theta_{adv}))$ for the utility of bidder $2$ when bidding truthfully faced with a truthful adversary of type $(g_{adv},\theta_{adv})$. We write $u^{attack(x,y)}_{(g_{adv},\theta_{adv})} = u^{\hat{\theta}=(1,1)}_{\theta'_1=(1,x), \theta'_2=(1,y) }((g_{adv},\theta_{adv}))$ for the utility of bidder $2$ when bidding the attack $(1,x),(1,y)$ faced with a truthful adversary of type $(g_{adv},\theta_{adv})$. 
We have
\begin{equation}
\label{eq:expUtilDiff_n=2}
\begin{split}
    & E^{q,F,n=2}_{\hat{\theta}=(1,\theta),truth}[u] - E^{q,F,n=2}_{\hat{\theta}=(1,\theta),attack(x,y)}[u] = \\
    & q \underbrace{E_{\theta_{adv} \sim F}[u^{truth}_{(1,\theta_{adv})} - u^{attack(x,y)}_{(1,\theta_{adv})}]}_{E1} + (1-q)\underbrace{E_{\theta_{adv} \sim F}[u^{truth}_{(2,\theta_{adv})} - u^{attack(x,y)}_{(2,\theta_{adv})}]}_{E2}
    \end{split}
    \end{equation}

By Lemma.~\ref{lem:q=1_analysis} we know that E1 is non-negative. If E2 is non-negative as well, then for every $q$ value the expression is non-negative and we're done. Otherwise, E2 must be negative, in which case the expression of Eq.~\ref{eq:expUtilDiff_n=2} is monotone in $q$. Thus, it's enough to show the BNE condition 
for $q=\frac{1}{2}$. Its correctness for higher $q$ values then follows by monotonicity.
For $q=\frac{1}{2}$ we thus require
\[
\begin{split}
& \frac{1}{2}E_{\theta_{adv} \sim F}[u^{truth}_{(1,\theta_{adv})} - u^{attack(x,y)}_{(1,\theta_{adv})} + u^{truth}_{(2,\theta_{adv})} - u^{attack(x,y)}_{(2,\theta_{adv})}] \geq 0.
\end{split}
\]
We now show that the inner expression of the expectation is non-negative for any value of $\theta_{adv}$. 
The following explicit expressions hold for the utilities: 

\[
\begin{aligned}
& u^{attack(x,y)}_{(2,\theta_{adv})} = 
      \begin{cases}
1 & 0 \leq \theta_{adv} \leq \frac{y}{2} \\
1 - 2\theta_{adv} + y & \frac{y}{2} < \theta_{adv} \leq \frac{x}{2} \\
1 - 4\theta_{adv} + y + x & \frac{x}{2} < \theta_{adv} \leq \frac{x + y}{2}  \\
0 & otherwise
\end{cases} \\[6mm]
    & u^{truth}_{(2,\theta_{adv})} = 
     \begin{cases}
1 - 2\theta_{adv} & 0 \leq \theta_{adv} \leq \frac{1}{2} \\
0 & otherwise
\end{cases}
\end{aligned} \hspace{0.5cm} 
\begin{aligned}
& u^{attack(x,y)}_{(1,\theta_{adv})} = 
 \begin{cases}
    1 - 2\theta_{adv} & 0\leq \theta_{adv} \leq y \\
    1-y & y < \theta_{adv} \leq 1 \\
    \end{cases} \\[16mm]
& u^{truth}_{(1,\theta_{adv})} = 1.
\end{aligned} \hspace{1cm}
\]

\noindent
Combining the utility functions we have (for $0\leq \theta_{adv}\leq 1$)
\[ 
\begin{split}
 u^{truth}_{(1,\theta_{adv})} & - u^{attack(x,y)}_{(1,\theta_{adv})} + u^{truth}_{(2,\theta_{adv})} - u^{attack(x,y)}_{(2,\theta_{adv})}  = \\
& \begin{cases}
1 - u^{attack(x,y)}_{(2,\theta_{adv})}   & 0\leq \theta_{adv} \leq \min\{y,\frac{1}{2}\} \\
6\theta_{adv} -1 - y - x  & \frac{1}{2} < \theta_{adv} \leq y \\
0 & y < \theta_{adv} \leq \frac{x}{2}\\
y + 1-2\theta_{adv}  & \frac{x+y}{2} < \theta_{adv} \leq \frac{1}{2}\\
 2\theta_{adv} - x  & \max\{\frac{x}{2},y\} < \theta_{adv} \leq \min\{\frac{x+y}{2},\frac{1}{2}\}\\
4\theta_{adv} - 1 - x  & \max\{y,\frac{1}{2}\}<\theta_{adv} \leq \frac{x+y}{2} \\
y  & \max\{\frac{x+y}{2},\frac{1}{2}\} < \theta_{adv} . \\
\end{cases}
\end{split}
\]

The first case expression is non-negative since $u^{attack(x,y)}_{(2,\theta_{adv})} \leq 1$ for whatever range $\theta_{adv}$ is in. 
The second case expression is non-negative since for this case we have $\theta_{adv}>\frac{1}{2}$, and generally $y\leq x \leq 1$. The fourth case expression is non-negative since for this case we have $\theta_{adv} < \frac{1}{2}$ and generally $y\geq 0$. The fifth case expression is non-negative since for this case we have $\theta_{adv} > \frac{x}{2}$. The 6th case expression is non-negative since for this case we have $\theta_{adv} > \frac{1}{2}$ and generally $x \leq 1$. The 7th case expression is non-negative since generally $y\geq 0$.  
\end{proof}

\begin{restatable}[]{lemma}{twotypeq}
\label{two_type_q*}
$q^*_{n=2,g=2} = \frac{2}{3}$.
\end{restatable}

The proof is similar to the one given in Lemma~\ref{lem:one_type_q*} and is given in Appendix C. 

\section{ Global Granularity Thresholds for $n>2$: $q^*_n = 1$ }
\label{sec:impossibility_result}
Based on the result of $q^*_{n=2} = \frac{2}{3}$ for two bidders, we may now ask what are the $q^*$ values for different values of $n$. We first give an intuitive example of why attacks can work well with $n\geq 3$. 

\begin{example}
Assume there are at least $\tilde{n} \geq 2$ adversaries, and all have per-item value $0.6$. Consider a true type bidder $n$ of type $(1,1)$, and the attack $(1,1), (1,0.5)$. We analyze two cases:

\begin{itemize}
\item \textbf{At least one of the adversaries is of type 1.} In this case the utility from truthful bidding and the attack is the same. The winning bids are $(1,1), (1,0.6)$. For the payment calculation, if $(1,1)$ is omitted, $(2,0.6)$ would be the winning bid, resulting in a social welfare of $1.2$. Therefore, the payment is $0.6$, the same as for truthful bidding. 

\item \textbf{All adversaries are of type 2.} In this case the utility from truthful bidding is 0 (as the bid would not be in the winning set) and the utility from the attack is 0.1, as $(1,1),(1,0.5)$ would be the winning bids, and the payments are $0.7, 0.2$, respectively. 
\end{itemize}

We conclude that the attack dominates truthful bidding under these conditions. When there is only one adversary with per-item value $0.6$, it can be directly verified that truthful bidding is a best response: in the first case, a truthful bidder wins and pays zero. In the second case, a truthful bidder would lose but the payment in any attack that makes the bidder a winner is larger than 1, which is the bidder's value for winning.
\end{example}

This example does not constitute a valid counter-example construction to the truthful BNE condition, as we have a true type and bids that are not in the support of the per-item value distribution (i.e., we assume all adversary bidders have per-item value $0.6$). If we allow a very small probability of the per-item value being either $0.5$ or $1$, this could constitute a valid counter-example. For this construction we need a more intricate analysis and a construction of $F$ that depends on the values of $q,n$. This is established by the following theorem:

\begin{theorem}
For any $q < 1, n > 2$, there is an attack $x,y$ and distribution $F$ such that the attack is beneficial when a bidder's true type is $(1,1)$. Therefore, for any $n>2$, $q^*_n =1$.
\end{theorem}
\label{thm:impossibility_n>=3}
\begin{proof}

We examine
the attack $(1,1),(1,\frac{1}{2})$ (so, $x=1, y=\frac{1}{2}$). 
We show how to choose $\epsilon(q,n)$ later, and define $F_{\epsilon(q, n)} = \begin{cases}
0.5 & w.p.~~\epsilon \\
0.6 & w.p.~~1 - 2\epsilon \\
1 & w.p.~~\epsilon \\
\end{cases}$

\noindent
Now fix some $q<1, n>2$. Consider the following series of mutually exclusive events (so, the later events assume that previous ones are already negated). Notice that their union covers all cases:

\begin{trivlist}

\item \textbf{C0: There is some adversary with per-item value of 0.5: } 
In this case bidding truthfully might yield higher expected utility than the above attack. A very rough bound shows that the added utility of bidding truthfully over the attack is at most $3$: the max utility is $1$ and minimal is $0$, and the min payment is $0$ and max payment is $2$ - so any two strategies have a utility difference of at most $3$.

$Pr[C0] = 1 - (1-\epsilon)^{\tilde{n}}$ (the complement of the event where all per-item values chosen i.i.d. from $F_{\epsilon(q,n)}$ are either 0.6 or 1). 

\item \textbf{C1: There exists at least one adversary of type 2 with value 1: } the utility for both the attack and truthful bidding is 0. The $(1,0.5)$ bid of the attack can not be included in the winning set. The $(1,1)$ bid (whether of the attack or the truthful bid) is charged a payment of 1 (and generates a value of 1) if it is included in the winning set by tie-breaking, and yields 0 value and payment if it is not. 

\item \textbf{C2: There exist at least two adversaries of type 1: } Since we're not under event $C0$, the two 1-type adversary bids have per-item values of at least $0.6$. In this case the attacker's bid $(1,\frac{1}{2})$ is not in the winning set. Moreover if the $(1,1)$ attack bid is in the winning set, then the price it pays is at least $0.6$, as there must be a $(1,0.6)$ adversary bid left out of the winning set. We conclude that both allocation and payments stay the same if we omit the $(1,0.5)$ attack bid, which is the same as truthful bidding, hence the utility for attack and truthful bidding is the same. 

\item \textbf{C3: There exists exactly one adversary of type 1: } Note that all adversaries of type 2 have a value of 0.6 (otherwise, we fall into the cases of $C0$ or $C1$). If the type 1 adversary has value 0.6, then both the attack and truthful bidding utilities are 0.4. We denote this event $C3A$. If the type 1 adversary has value 1, then the attack utility is 0.5 while truthful utility is 0.8. We denote this event $C3B$. 

$Pr[C3B] = \overbrace{\binom{\tilde{n}}{1} (1-q)^{\tilde{n}-1}q}^{Pr1} \overbrace{(1-2\epsilon)^{\tilde{n}-1} \epsilon}^{Pr2}$. $Pr1$ is the probability that exactly one adversary is of type $1$, and $Pr2$ is the probability that all type $2$ bidders have per-item valuation of $0.6$, and the one type $1$ bidder has per-item valuation of $1$.

\item \textbf{C4: There are no adversaries of type 1: } we are left with only type 2 adversaries that have value 0.6. In this case, the truthful utility is 0 and the attack utility is 0.1. 

$Pr[C4] = (1-q)^{\tilde{n}}(1-2\epsilon)^{\tilde{n}}$, the probability to choose only type $2$ bidders, all with per-item valuation of $0.6$. 

\end{trivlist}

We now choose $$\epsilon(q,n) = \min \{ \frac{1-q}{3000q\tilde{n}}, \frac{1}{1000}, \frac{1 - (\frac{1}{5})^{\frac{1}{\tilde{n}}}}{2}, 1 - (1 -  \frac{(1-q)^{\tilde{n}}}{300})^{\frac{1}{\tilde{n}}} \}.$$ All these expressions are strictly positive for $n > 2, q< 1$. 
For the difference in expected utilities between bidding the truth and the attack, we have:

\[
\begin{split}
    & E^{q,F_{\epsilon(q,n)},n}_{(1,1),Truth} - E^{q,F_{\epsilon(q,n)},n}_{(1,1),attack(1,\frac{1}{2})} \leq 3Pr[C0] + 0.3 Pr[C3B]  - 0.1 Pr[C4] = \\
    & 3(1-(1-\epsilon)^{\tilde{n}}) + 0.3\tilde{n}(1-2\epsilon)^{\tilde{n}-1}(1-q)^{\tilde{n}-1}q\epsilon - 0.1 (1-q)^{\tilde{n}}(1-2\epsilon)^{\tilde{n}} \stackrel{\epsilon \leq \frac{1-q}{3000q\tilde{n}}}{\leq} \\
    & 3(1-(1-\epsilon)^{\tilde{n}}) + 0.1(1-2\epsilon)^{\tilde{n}-1}(1-q)^{\tilde{n}}\frac{1}{1000} - 0.1 (1-q)^{\tilde{n}}(1-2\epsilon)^{\tilde{n}} \stackrel{\epsilon \leq \frac{1}{1000}}{\leq} \\
    & 3(1-(1-\epsilon)^{\tilde{n}}) + 0.1(1-2\epsilon)^{\tilde{n}}(1-q)^{\tilde{n}}\frac{1}{998} - 0.1 (1-q)^{\tilde{n}}(1-2\epsilon)^{\tilde{n}} < \\
    & 3(1-(1-\epsilon)^{\tilde{n}}) - 0.05 (1-q)^{\tilde{n}}(1-2\epsilon)^{\tilde{n}} \stackrel{\epsilon \leq \frac{1 - (\frac{1}{5})^{\frac{1}{\tilde{n}}}}{2}}{\leq} \\
    & 3(1-(1-\epsilon)^{\tilde{n}}) - 0.01(1-q)^{\tilde{n}} \stackrel{\epsilon \leq 1 - (1 -  \frac{(1-q)^{\tilde{n}}}{300})^{\frac{1}{\tilde{n}}}}{\leq} 0 
    \end{split}
    \]
    
    Notice that one of the inequalities is strict. We conclude that the attack is beneficial and truthful bidding is not a BNE. 

\end{proof}



\section{Beta Distributions Granularity Thresholds: a Computer Algebra Approach}
\label{sec:beta_distributions}


Given the result of the previous section, we know that it is impossible to guarantee a truthful BNE for arbitrary distributions when $q < 1$ and $n \geq 3$. In this section we instead consider Beta distributions, parameterized by $\alpha, \beta$ we assume are integers. This family of distributions includes many natural distributions, such as the uniform distribution and other distributions with CDF $F(x)=x^a$ for any integer $a \geq 1$. It can additionally be used to polynomially approximate arbitrary distributions. Formally, a probability density function in this family has the form $f_{\alpha, \beta}(\theta_i) = (\alpha+\beta-1)! \frac{\theta_i^{\alpha - 1}(1-\theta_i)^{\beta - 1}}{(\alpha - 1)!(\beta - 1)!}$ and we denote by $F_{\alpha, \beta}(\theta_i)$ the respective cumulative distribution function.

Our aim is to find the corresponding granularity thresholds  $q^*_{n,F_{\alpha, \beta}}$.     
We address this challenge as follows. We show that the difference between the expected utility of a 1-type truthful bidder (with $\theta = 1$ w.l.o.g.) and the expected utility of an attacker with $attack(x,y)$ is a polynomial $P(x,y,q)$. Similarly, we show that the difference between the expected utility of a 2-type truthful bidder with some per-item value $\theta$ and the expected utility of an attacker with $attack(x,y)$ is a polynomial $Q(\theta,x, y,q)$.
The minimal $q$ for which the polynomials are non-negative for the entire domain $[0,1]\times[0,1]\times[q,1]$, is in fact $q^*_{n,F_{\alpha, \beta}}$. 
This reduces the problem into a computer algebra problem, given parameters $n,\alpha,\beta$ which can be solved using state-of-the-art techniques. The conceptual message that we wish to convey with this result is that the impossibility of the previous section becomes a possibility when one restricts attention to a natural family of distributions.


\subsection{A reduction to a polynomials positivity decision problem} 

The goal of this subsection is to show that for a given beta distribution $F$ and number of bidders $n$ we can reduce the problem of finding $q^*$ to the problem of showing that certain polynomials are positive for any variable assignment. We first simplify the analysis by showing that the expression for bidder $n$'s utility depends only on the top adversary bids. This means that instead of sampling all $\tilde{n}$ adversaries from $F$, it is enough to sample the top adversary bids from the appropriate order statistic of $F$. Recall that we denote by $v_1, v_2$ the highest two 1-type adversary per-item values and by $w_1$ the highest 2-type per-item value.

\begin{claim}
Bidder $n$'s utility is determined by her true type, her own bids and adversary bids $v_1,v_2,w_1$. 
\end{claim}
\begin{proof}
Bidder $n$'s utility is determined by whether it has bids in the winner set, and if it does, by the payment it is required to pay.

\begin{itemize}
\item A 2-type bid of bidder $n$ enters $W(B)$ iff $w_1' > M$, with $M = \max \{ v_1 + v_2, v_1' + v_2', v_1 + v_1', w_1 \}$. The payment is then $M$. 

\item Exactly one 1-type bid of bidder $n$ enters $W(B)$ iff $v_1' > v_2$ and $v_1' + v_1 > \max \{ w_1, w_1'\}$. The payment is then $\max \{v_2, v_2', \max\{w_1, w_1'\} - v_1\}$

\item Two 1-type bids of bidder $n$ enters $W(B)$ iff $v_2' > v_1$ and $v_1' + v_2' > \max \{ w_1, w_1'\}$. The payment is then $\max \{w_1, w_1', v_3' + v_4', v_3' + v_1, v_1 + v_2\}$ 

\end{itemize}

\end{proof}

We are thus justified in defining  $$\tilde{u}_{S}^{\hat{\theta}}(w_1,v_1,v_2) = u^{\hat{\theta}}_{S}(\hat{\theta_1}, \ldots, \hat{\theta}_{\tilde{n}}).$$

We now wish to develop more explicit expressions for $E_{\hat{\theta},S}^{k, F_{\alpha, \beta}, n}$ (as defined in Eq.~\eqref{eq:expectedUtilK}). We do so for $k\geq 2, \tilde{n}-k\geq 1$, but adjusted (and simpler) expressions can be developed for the cases of $k \in \{ 0,1,\tilde{n} \}$. We note that the probability of sampling $k$ i.i.d. samples from $F_{\alpha, \beta}$ resulting in top 1-type values $v_1,v_2$, is the same as pre-choosing the top samples $v_1,v_2$, and having the rest of the samples being of lower value (with probability $F_{\alpha,\beta}(v_2)^{k-2}$). When $k\geq 2$, there are $k(k-1)$ ways of choosing which of the $k$ samples are the top samples. A similar analysis holds for the probability of sampling a top 2-type value $w_1$. We have


\begin{equation}
\begin{split}
\label{eq:expected-utility-k}
        & E_{\hat{\theta},S}^{k, F_{\alpha, \beta}, n} = \\
        & k(k-1)\int_{v_1=0}^{1} f_{\alpha,\beta}(v_1)\int_{v_2=0}^{v_1} f_{\alpha,\beta}(v_2) F_{\alpha, \beta}(v_2)^{(k-2)} (\tilde{n}-k)\int_{w_1 = 0}^{1} f_{\alpha, \beta}(w_1) F_{\alpha, \beta}^{\tilde{n}-k-1}(w_1) \tilde{u}^{\hat{\theta}}_{S}(w_1,v_1,v_2) dw_1 dv_2 dv_1
        \end{split}
\end{equation}

The utility function $\tilde{u}$ has the general structure (the details can be found in Appendix D, and depend on $\hat{\theta}, S$) where there is some finite number of segments partitioning the integration bounds of $0 \leq v_2 \leq v_1 \leq 1, 0\leq w_1 \leq 1$. Each of these segments is of the form $\underbar{T}_1(x,y,\theta) \leq v_1 \leq \bar{T}_1(x,y,\theta), \underbar{T}_2(v_1,x,y,\theta) \leq v_2 \leq \bar{T}_2(v_1,x,y,\theta), \underbar{T}_3(v_1,v_2,x,y,\theta) \leq w_1 \leq \bar{T}_3(v_1,v_2,x,y,\theta)$, where all $T$ functions are linear in their variables. For each of these segments, the utility function $\tilde{u}$ is linear in its variables. Notice that $F_{\alpha, \beta}(x), f_{\alpha, \beta}(x)$ are polynomial in their variables as well. Overall, 
there is some finite $t$ so we can write (omitting the variables that the T functions receive), 
\[
\begin{split}
 & k(k-1)\int\limits_{v_1=0}^{1} f_{\alpha,\beta}(v_1)\int\limits_{v_2=0}^{v_1} f_{\alpha,\beta}(v_2) F_{\alpha, \beta}(v_2)^{(k-2)} (\tilde{n}-k)\int\limits_{w_1 = 0}^{1} f_{\alpha, \beta}(w_1) F_{\alpha, \beta}^{\tilde{n}-k-1}(w_1) \tilde{u}^{\hat{\theta}}_{S}(w_1,v_1,v_2) dw_1 dv_2 dv_1 = \\
 & \sum_{i=1}^t k(k-1)\int\limits_{v_1=\underbar{T}_1^i}^{\bar{T}_1^i} f_{\alpha,\beta}(v_1)\int\limits_{v_2=\underbar{T}_2^i}^{\bar{T}_2^i} f_{\alpha,\beta}(v_2) F_{\alpha, \beta}(v_2)^{(k-2)} (\tilde{n}-k)\int\limits_{w_1 = \underbar{T}_3^i}^{\bar{T}_3^i} f_{\alpha, \beta}(w_1) F_{\alpha, \beta}^{\tilde{n}-k-1}(w_1) \tilde{u}^{\hat{\theta}}_{S}(w_1,v_1,v_2) dw_1 dv_2 dv_1
\end{split}
\]

Now since the class of polynomials is closed under integration and composition, we have that $E_{\hat{\theta},S}^{k, F_{\alpha, \beta}, n}$ is polynomial in $x,y,\theta$. For the case of a 1-type true type, Theorem~\ref{thm:attack_format} shows we can assume w.l.o.g. $\theta = 1$. 
We can then write 
\[
\begin{split}
& P_{\alpha, \beta}^{n}(x,y,q) \stackrel{def}{=} E^{q, F_{\alpha,\beta}, n}_{(1,1),truth}[\tilde{u}]\!{-}\! E^{q, F_{\alpha,\beta}, n}_{(1,1),attack(x,y)}[\tilde{u}] = \\
& \sum_{k=0}^{\tilde{n}} {\binom{\tilde{n}}{k}} q^k (1-q)^{\tilde{n}-k}(E_{(1,1),truth}^{k, F_{\alpha, \beta}, n}[\tilde{u}] - E_{(1,1),attack(x,y)}^{k, F_{\alpha, \beta}, n}[\tilde{u}]),
\end{split}
\]
\[
\begin{split}
& Q_{\alpha, \beta}^{n}(\theta,x,y,q) \stackrel{def}{=} E^{q, F_{\alpha,\beta}, n}_{(2,\theta),truth}[u]\!{-}\!E^{q, F_{\alpha,\beta}, n}_{(2,\theta),attack(x,y)}[u] = \\
& \sum_{k=0}^{\tilde{n}} {\binom{\tilde{n}}{k}} q^k (1-q)^{\tilde{n}-k}(E_{(2,\theta),truth}^{k, F_{\alpha, \beta}, n}[\tilde{u}] - E_{(2,\theta),attack(x,y)}^{k, F_{\alpha, \beta}, n}[\tilde{u}]), \\
\end{split}
\]

and these as well are polynomials in their variables, as they have a representation as a sum of multiplications of polynomials in their variables. If for any valid variable assignment these polynomials are positive, then the truthful BNE condition holds for any true type, and vice versa. Thus,

\begin{theorem}
\label{thm:BetaPolynomialReduction}
For parameters $q,n$ and a beta distribution $F_{\alpha, \beta}$, truthful bidding is a Bayesian Nash Equilibrium (BNE) if and only if it holds that the polynomials 
\[
\begin{split}
& P_{\alpha, \beta}^{n}(x,y,q) \geq 0 \\
& Q_{\alpha, \beta}^{n}(\theta, x, y,q) \geq 0 \\
\end{split}
\]

for any variable assignment. 
\end{theorem}

\subsection{Solving polynomial positivity using computer algebra} 
\label{subsec:computer_algebra_methods}

The question of whether truthful bidding is a Bayesian Nash equilibrium for a given $q$ value is now reduced to whether there are no assignments (a true type and attack bids) where the polynomial is negative in the domain.
For this purpose, any computer algebra method that is able to prove polynomial positivity suffices. The method we found most useful for our setting is Partial Cylindrical Algebraic Decomposition (Partial CAD) \cite{CADBook}. Given a guess of the threshold value $q^*\_guess$ for a given $F_{\alpha,\beta}$ distribution and $n$ value, Maple's Regular Chains package \cite{RegularChains} has the function \emph{PartialCylindricalAlgebraicDecomposition} that given the proper constraints finds representative points in the domain $[0,1]^2\times[q^*\_guess,1]$ for the positivity decision problem. 
It then suffices to check each of the representing points to be positive for the corresponding polynomial in order to conclusively verify that it does not attain negative values in the entire domain. If we also know that for any value $q<q^*\_guess$ there are assignments of the other variables such that the polynomial is negative, then we establish $q^* = q^*\_guess$. 

One way to arrive at the guess value $q^*\_guess$ is by bisecting the $[0,1]$ interval. However, this is computationally costly since every bisecting step needs to execute the Partial CAD procedure for a domain of three variables.

It is also possible to fix some variable assignment (truth value and attack), find the minimal $q$ for which this polynomial in $q$ (with fixed true value and attack) is non-negative, and use this as $q^*\_guess$. By doing so we already establish that for any $q < q^*\_guess$ there are assignments where the polynomial is negative, namely the assignment we fixed. As described above, all that is left is to use the Partial Cylindrical Algebraic Decomposition procedure to verify it. 
We find that for any $\alpha, \beta, n$ parameters we checked, the ``split attack'', where the true type is $(2,1)$ and the attack is $(1,1), (1,1)$, gives a proper $q^*\_guess$ value. 
We show a result for the split attack in the next section. 

Our findings 
are given in Figure~\ref{fig:different_alpha_n}. We emphasize that these are true (rather than estimated) $q^*$ values, as the methods described constitute a full proof for the specific values tested. For computational reasons, we restrict the tests to $\beta = 1$. This enables use of Lemma~\ref{lem:beta=1_assumptions}, hence fixing some variables values without loss of generalization, and thus reducing the algebraic complexity. As can be seen, the granularity thresholds in all cases we tested are bounded away from 1, in particular, lower than $0.5$.

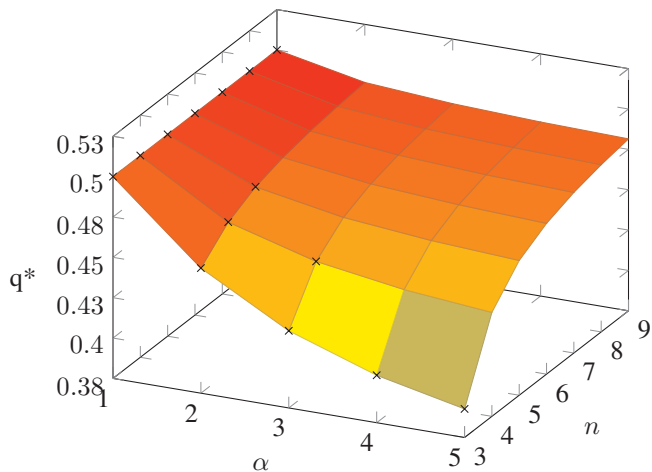
\begin{figure}[!h]
\begin{tikzpicture}
\begin{axis}[
    xlabel={$\alpha$},
    ylabel={$n$},
    zlabel={q*},
    zlabel style={rotate=-90},
    xmin=1, xmax=5,
    ymin=3, ymax=9,
    zmin=0.375, zmax=0.525,
    symbolic x coords={1,2,3,4,5}, 
    xticklabels={1,2,3,4,5},
    xtick={1,2,3,4,5},
    symbolic y coords={3,4,5,6,7,8,9}, 
    yticklabels={3,4,5,6,7,8,9},
    ytick={3,4,5,6,7,8,9},
    ztick={0.375,0.4,0.425,0.45,0.475,0.5,0.525},
    grid style=dashed,
]
\addplot3[
    surf,
    domain=1:5,y domain=3:9,
] 
coordinates {
(1,3,0.5) (1,4,0.5) (1,5,0.5) (1,6,0.5) (1,7,0.5) (1,8,0.5) (1,9,0.5)
 
(2,3,0.4525) (2,4,0.4679) (2,5,0.4767) (2,6,0.4818) (2,7,0.4851) (2,8,0.4874) (2,9,0.4890)

(3,3,0.4228) (3,4,0.4527) (3,5,0.4669) (3,6,0.4746) (3,7,0.4794) (3,8,0.4826) (3,9,0.4850)

(4,3,0.4044) (4,4,0.4444) (4,5,0.4617) (4,6,0.4708) (4,7,0.4763) (4,8,0.4801) (4,9,0.4828)

(5,3,0.3927) (5,4,0.4394) (5,5,0.4585) (5,6,0.4684) (5,7,0.4745) (5,8,0.4786) (5,9,0.4815)
};

    \addplot3[domain=4:5, y domain=1:10, mark=x, color=black, only marks]
    coordinates {
    (1,3,0.5) (1,4,0.5) (1,5,0.5) (1,6,0.5) (1,7,0.5) (1,8,0.5) (1,9,0.5) (2,3,0.4525) (2,4,0.4679) (2,5,0.4767) (3,3,0.4228) (3,4,0.4527) (4,3,0.4044) (5,3,0.3927)};
\end{axis}
\end{tikzpicture}
\caption{q* values for beta distributions with $\beta = 1$}
\label{fig:different_alpha_n}
\end{figure}

Note that the fully verified cases are those marked with $x$ symbol. The cases not marked with $x$ give the $q^*$ values only for the split attack, without verifying that they hold for all possible attacks. For the fully verified cases, we observe that the granularity threshold is typically decreasing (improving) as a function of $\alpha$ and increasing as a function of $n$. As mentioned before, interestingly, in all fully verified cases, the attack most persistent in respect to higher values of $q$ (``best attack'') was the split attack. We conjecture that this phenomenon is universal for the family of $(\alpha, \beta)$ distributions (though it is not necessarily true for other distributions, see Section~\ref{sec:Discussion}). We therefore extend the figures to include $q^*$ values that were evaluated only for the split attack.

The Mathematica and Maple files to attain and verify the figures' exact values independently can be found at https://github.com/yotam-gafni/vcg\_bayesian\_fnp.

\section{The Uniform Distribution and Split Attacks}
\label{sec:split_attack}

Fig.~\ref{fig:different_alpha_n} may suggest the conjecture that $q^* \leq 0.5$ whenever $\beta=1$. A special case is the uniform distribution for which $\alpha=\beta=1$. In this section we examine this special case. The conjecture is related to another empirical observation which may be interesting in its own right, regarding the split attack: {\em all} the $q^*$ values given in the figures of the last section are in fact derived from the split attack which proves itself in all verified cases to be the one that yields the appropriate $q^*\_guess$ value (see the discussion in the beginning of Section~\ref{subsec:computer_algebra_methods}). This raises a second conjecture, that the split attack obtains the maximal expected utility among all attacks, for a 2-type, and when the distribution is uniform.

Recall that in the split attack that we have previously considered, a 2-type attacker with value 1 splits her bid into two 1-type bids $x=y=1$. Here we consider a more general definition where a 2-type attacker with value $\theta$ splits her bid into two 1-type bids $x=y=\theta$.

\begin{theorem}
\label{thm:uniform_half}
For the uniform distribution, and any number of bidders $n\geq 3$,
\begin{enumerate}
    \item $q^*_{n,g=1,F=UNI([0,1])}\leq \frac{1}{2}$, i.e., $q^*$ for a 1-type and adversaries drawn from the uniform distribution is at most $\frac{1}{2}$.
    \item $q^*_{n,g=2,\theta,x=\theta,y=\theta,F=UNI([0,1])}\leq\frac{1}{2}$, which is tight as $q^*_{n,g=2,\theta=1,x=1,y=1,F=UNI([0,1])}=\frac{1}{2}$. I.e., $q^*$ for a 2-type performing a split attack, and adversaries drawn from the uniform distribution, is at most $\frac{1}{2}$.
\end{enumerate}
\end{theorem}

\noindent
The proof is given in Appendices~\ref{app:UniformDistProof} and~\ref{app:uniform_half_1type}. We further believe that the threshold of $\frac{1}{2}$ holds generally for the uniform distribution:

\begin{conjecture}
\label{conj:uniform_half}
~
\begin{enumerate}
    \item For a player with type $(g=2,\theta)$, assuming that $F$ is the uniform distribution, the maximal expected utility among all attacks. 
    Formally,
    
    $$ E_{\hat{\theta}=(2,\theta), attack(\theta, \theta)}^{q,n} \geq E_{\hat{\theta}=(2,\theta), attack(\frac{x+y}{2},\frac{x+y}{2})}^{q,n} \geq E_{\hat{\theta}=(2,\theta), attack(x,y)}^{q,n}.$$
    
    \item For the uniform distribution, and any $n\geq 3$, $q^*_{n,F=UNI([0,1])}=\frac{1}{2}$.
\end{enumerate}
\end{conjecture}

\noindent
The first item of this conjecture implies the second item (using Theorem~\ref{thm:uniform_half}). We verified the second item of this conjecture for the missing case of $g=2$ and non-split-attacks up to $n=9$ using the methods of section~\ref{sec:beta_distributions}, as shown in Figure~\ref{fig:different_alpha_n}. 

While we conjecture that the split attack yields the highest $q^*$ when $F$ is the uniform distribution, for other distributions this conjecture may not be correct, as the following counter-example demonstrates.\footnote{It is also incorrect that for all distributions the split attack is the best attack for a 2-type. We omit the counter-example for brevity of the exposition and are happy to share it upon request.}

\begin{example}
\label{ex:split_attack_dominated}
A distribution where an attack has higher $q^*$ value than the split attack. Specifically, we show that
$q^*_{n=3,g=1,\theta=1,x=1,y=0.2,F} > q^*_{n=3,g=2,\theta=1,x=1,y=1,F}$
for the distribution $$F = \begin{cases} 1 & \text{w.p. }0.5 \\
0.1 & \text{w.p. }0.5 \end{cases}$$

\noindent
(Note that one of the attack values is not in the distribution support. This is for the simplicity of the calculation and can be fixed by adding this value to the support, with a very small probability, similar to the proof of Theorem~\ref{thm:impossibility_n>=3}).

\vspace{5mm}

\noindent
{\bf Calculation of $q^*_{n=3,g=2,\theta=1,x=1,y=1,F}$:}
For truthful bidding, we have
$$E^{q,F,n=3}_{\hat{\theta}=(2,1),truth} = q^2 (\frac{1}{2} 0.9 + \frac{1}{4} 1.8) + 2q(1-q) (\frac{1}{4} \cdot 1 + \frac{1}{4} \cdot 1.8) + (1-q)^2 (\frac{1}{4}\cdot 1.8) $$

For the split attack, we have
$$ E^{q,F,n=3}_{\hat{\theta}=(2,1),attack(1,1)} = q^2(\frac{1}{4} \cdot 1.8) + 2q(1-q)(\frac{1}{4} \cdot 1.8) + (1-q)^2(\frac{1}{4} \cdot 2)$$

For truthful bidding to be a BNE we need: $$E^{q,F,n=3}_{\hat{\theta}=(2,1),truth} - E^{q,F,n=3}_{\hat{\theta}=(2,1),attack(1,1)} = q^2(\frac{1}{2}\cdot 0.9) + 2q(1-q)(\frac{1}{4}\cdot 1) - (1-q)^2(\frac{1}{4} \cdot 0.2) \geq 0$$ 
which holds if and only if $q \geq \frac{1}{2}(6 - \sqrt{34}) \approx 0.0845241$. 

\vspace{5mm}

\noindent
{\bf Calculation of $q^*_{n=3,g=1,\theta=1,x=1,y=0.2,F}$:}
For truthful bidding, we have
$$E^{q,F,n=3}_{\hat{\theta}=(1,1),truth} = q^2(\frac{3}{4} 0.9) + 2q(1-q) (\frac{1}{4} \cdot 1 + \frac{1}{4} \cdot 0.9) + (1-q)^2 (\frac{1}{4} \cdot 0.8)$$

For the attack $(1,1),(1,0.2)$, we have
$$ E^{q,F,n=3}_{\hat{\theta}=(1,1),attack(1,0.2)} = q^2(\frac{3}{4} \cdot 0.8) + 2q(1-q) (\frac{1}{2} \cdot 0.8) + (1-q)^2 (\frac{1}{4} \cdot 1)$$

For truthful bidding to be a BNE we need:
$$E^{q,F,n=3}_{\hat{\theta}=(1,1),truth} - E^{q,F,n=3}_{\hat{\theta}=(1,1),attack(1,0.2)} = q^2(\frac{3}{4} 0.1) + 2q(1-q)(\frac{1}{4} \cdot 0.3) - (1-q)^2 (\frac{1}{4} \cdot 0.2) \geq 0 $$

which holds if and only if $q \geq \frac{1}{5} (5 - \sqrt{15}) \approx 0.225403$. 

\end{example}

\section{Discussion}
\label{sec:Discussion}

In this work, we investigate the most simple model fit for a Bayesian analysis of VCG under Sybil attacks. We show both possibility and impossibility results.
Our results are given in terms of one specific property of the underlying distribution of types -- the probability $q$ that a player demands one item. If $q=1$ and all players demand one item, truthful bidding in VCG is known to be a dominant strategy even when false-name attacks are allowed. Intuitively, as $q$ gets closer to $1$ we would expect truthful bidding to be more robust to false-name attacks. We study the ``granularity threshold'', namely the minimal $q^*$ (given a specific setting) such that for every $q>q^*$, truthful bidding is a BNE in this setting. 

We obtain both possibility and impossibility results. Our main negative results shows that $q^*_n=1$ for any number of players $n \geq 3$. I.e., for any $q<1$ there exists a per-item value distribution such that truthful bidding in the VCG auction is not a BNE. The simplicity of our model makes this impossibility result more robust. Our main positive result shows that $q^*_{n=2}=\frac{2}{3}$. This means that for two players, if the probability of a 1-type is at least $\frac{2}{3}$ then truthful bidding is a BNE for any per-item value distribution.

Our impossibility result is a worst-case result and we believe that for many natural distributions $q^*_n$ is significantly smaller than 1 even for $n>2$. To demonstrate this issue we study the class of beta distributions, and show that $q^*$ is indeed significantly smaller than 1 for several such distributions using a computer algebra approach. In particular, we conjecture that $q^*_n \leq \frac{1}{2}$ for any $n$ and any beta distribution with $\beta=1$ (and any $\alpha$). One such example is the uniform distribution and for this case we give some supporting evidence in Section~\ref{sec:split_attack}. We believe that this is a challenging mathematical question that we leave open for future research. The discussion in Section~\ref{sec:split_attack} points attention to a specific kind of attack that we term the split attack, and it will be additionally interesting to further understand the conjecture in that section and the properties of this attack and its connection to $q^*$. 

Beyond the uniform distribution, it is interesting to determine whether a formula can be derived for the asymptotic behavior of $q^*$ as a function of $\alpha, \beta$ for beta distributions, or for specific $n$ values. Can we find larger classes of distributions that admit a nice analysis? Can we exploit approximations of other general distributions by Beta (or more generally polynomial) distributions to yield precise bounds for their respective $q^*$ values?  

While our model is certainly a special case of the general model of combinatorial auctions, its simplicity is misleading since the mathematical challenges that it poses are far from trivial, and the conceptual conclusions that it yields are far from obvious (in fact they sometime contradict basic intuition). We next discuss some of these conclusions, and the complexities and mathematical obstacles that hide underneath the misleadingly simple definition of model, which prevents a simple decomposition of the problem. We divide this discussion to two parts. We first discuss some intuitively desirable structural properties that are either incorrect or hard to prove. We then discuss some consequences of two possible generalizations of our model.

\subsection{Some desirable structural properties of our model}
\label{sec:structural_properties}
\begin{trivlist}

\item {\bf Is truthfulness monotone in $q$? } I.e, if an attack is beneficial for some $q$, does this imply that it must be beneficial for any $q' < q$. If this is the case, $q^*$ has a clear cut meaning: For any $q > q^*$, truthful bidding is a BNE, and for any $q< q^*$, it is not. This property is intuitive because, as we observed, generally an attack is more helpful while facing 2-type adversaries. We do not know if this property generally holds, even for $n=3$. We do know it is true for the split attack (a proof is given in Lemma~\ref{lem:split_is_monotone_k} as part of the proof of Theorem~\ref{thm:uniform_half}). One proof technique for general attack formats could be to show that, for any adversary profile, ``flipping'' one of the adversaries from a 2-type to a 1-type only helps truthful bidding. However quite surprisingly this is not true:

\theoremstyle{remark}
\newtheorem{example2}[theorem]{Example}
\begin{example2}
\label{ex:truth_q_monotonicity_attempt}
It's not always better for an attacker to face 2-types. Consider a type $(1,1)$ attacker with $attack(0.1,0.1)$, and adversarial bid $(1,0.15)$. The attacker is worse-off than bidding truthfully by $0.15$. But if the adversarial bid is $(2,0.15)$, then the attacker is worse-off by $0.7$. 
\end{example2}

\item {\bf Is truthfulness monotone in $g$? } I.e., if an attack $x,y$ is beneficial for some 1-type bidder in some setting $q,F,\theta$, then it is beneficial for a 2-type bidder in the same setting. Formally, $q^*_{n,g=2,\theta,x,y,F} \geq q^*_{n,g=1,\theta,x,y,F}$. This is intuitive since 1-type attackers bid for more items than they actually require, which seems risky. Moreover, we find empirically this inequality holds in the tested settings. Yet this conjecture is false, as can be extracted from Example~\ref{ex:split_attack_dominated}, where $q^*_{n=3,g=1,\theta=1,x=1,y=0.2,F} > q^*_{n=3,g=2,\theta=1,x=1,y=0.2,F}$ (the calculation of the R.H.S.~expression is not included in the example but can be derived directly). 

\item {\bf Is $q^*_{F(alpha,beta)}$ monotonically decreasing in $\alpha$, even just for $\beta=1$?} An empirical observation out of Figure~\ref{fig:different_alpha_n} is that for $\beta = 1$, $q^*_{F_{\alpha, \beta = 1}}$ is monotone in $\alpha$. We do not know if this is generally correct.

\item {\bf Is $q^*$ monotone in $n$?} As discussed in the introduction, both for $q^*_n$ (Theorem~\ref{thm:impossibility_n>=3}) and for $q^*_{n,F}$ for some given distributions $F$ (Figure~\ref{fig:different_alpha_n}), it seems $q^*$ is monotonically \emph{increasing} in $n$, which we believe is surprising.

\end{trivlist}

\subsection{The necessity of some of our assumptions}

\begin{trivlist}

\item {\bf General type distributions. }
Our model assumes $F \stackrel{def} = F_1=F_2$, i.e., the same per-item value distribution for different item demand types. When $F_1 \neq F_2$ our main positive result for $n=2$ does not generally hold as the following example demonstrates.

\begin{example}
\label{ex:F1F2}
For $n=2$ and for all $q<1$ there are distributions $F_1(q), F_2(q)$ such that truthful bidding is not a Bayesian-Nash Equilibrium. Consider the following distributions:

$$F_1 = \begin{cases}
   0.5 & w.p. \epsilon(q) \\
   \epsilon(q) & w.p. 1 - \epsilon(q) \\
\end{cases}, F_2 = 0.2 w.p. 1.$$

\noindent
I.e., a 1-type bidder has per-item value $\epsilon(q)$ with high probability, and a 2-type bidder always has a per-item value of $0.2$. Now fix some $q< 1$, choose $\epsilon(q) = \min \{0.1\frac{1-q}{q},1 \}$, and consider bidder $2$ with true type $(2,0.2)$. If she bids an attack $(1,0.5),(1,0.5)$:

\begin{itemize}

\item If faced with a $(2,0.2)$ adversary, attacker pays nothing, instead of paying $0.4$. 

\item If faced with a $(1,\epsilon)$ adversary, attacker pays $2\epsilon$ instead of $\epsilon$. 

\item If faced with a $(1,0.5)$ adversary, attacker gets the item (worth $0.4$ by her valuation) and pays $1$, rather than not winning in case of truthful bid. 

\end{itemize}

\noindent
Her expected utility is therefore: 
\[
\begin{split}
    E^{q,F_1,F_2,n=2}_{(2,0.2),attack(0.5,0.5)}[u] & - E^{q,F_1,F_2,n=2}_{(2,0.2),truth}[u] =
    (1-q)0.4 - q(1-\epsilon)\epsilon - q\epsilon 0.5 = \\
    & (1-q)0.4 - (1-q)0.1(1-\epsilon) - 0.05(1-q) \geq 0.25(1-q) > 0.
    \end{split}
    \]

\end{example}

\noindent
This counter-example has the property that the expected value of a 1-type per-item valuation must decrease with $q$ for the attack to succeed. An interesting future research question is whether Theorem~\ref{thm:global_granularity_n=2} still holds in some capacity when the per-item values are not identical but rather bound together by a looser correlation or stochastic dominance criterion.

\item {\bf More than two items.}
Given the computer algebra methods for the beta distribution it is possible to discuss a larger number of items. In the two item case we were able to prove Lemma~\ref{lem:attack_format} that shows a bidder has two alternatives - to submit a truthful bid or submit a pair $(1,x),(1,y)$ with $0\leq y \leq x \leq 1$. This results in a tri-variate polynomial with q,x,y. With more items (e.g., 3) one might consider other scenarios: The bidder submits (3,x), or $(2,x),(1,y)$, etc - but still a small finite space that depends on the number of items. For each case one derives a polynomial as done in section 5, and using the computer algebra solvers can come up with analytic results. Though computationally more demanding, in theory this could extend to any number of items. Notice that in the two item case, there is a single number q that measures ‘granularity’ by quantifying the probability of a 1-item demand bidder to appear. Once we analyze 3 items or more, one needs choose a finer way of defining granularity - for example, we could say that having the probability vector (0.4,0.5,0.1) over the amounts (1,2,3) of item demand is more granular than (0.4,0.1,0.5), even though the “q” value as defined for both is 0.4.

\end{trivlist}

\section{Acknowledgements}
We wish to thank Doron Zeilberger and Christoph Koutschan for their kind advise regarding symbolic identity proving. We also wish to thank Noam Neer for his advise regarding polynomial positivity proof techniques.

Yotam Gafni and Moshe Tennenholtz were supported by the European Research Council (ERC) under the European Union’s Horizon 2020 research and innovation programme (Grant No. 740435).  

Ron Lavi was partially supported by the ISF-NSFC joint research program (grant No. 2560/17).

\appendix

\section{Technical Assumptions for Analysis}

The following section constitutes a proof for Theorem~\ref{thm:attack_format} in Section 2:
\attackformat*

Lemma~\ref{lem:attack_format} shows that we can restrict attention to attacks of the form $(1,x),(1,y)$. Lemmas~\ref{lem:vcg_homogenous},\ref{lem:distribution_assumption} show that when considering $q^*_n$, we can w.l.o.g. assume either $\theta = 1$ or $x=1$, as for any distribution $F$ and parameters $x,\theta$ such that an attack is beneficial, we can construct a distribution $F'$ with either $\theta = 1$ or $x=1$ and an attack that is beneficial. Lastly, we show that for 1-type bidders it is w.l.o.g. to assume $x\leq \theta$, and so combined with the previous conclusion we may assume $\theta = 1$ when analyzing $q^*_n$ for 1-types.  

\begin{lemma}
\label{lem:attack_format}
For any attack of bidder $n$, $S = \hat{\theta_1'},...,\hat{\theta_m'}$,
one of the following must be true:
\begin{itemize}
\item There exists an attack $(1,x),(1,y)$ such that for any adversary bids, bidder $n$'s expected utility from $(1,x),(1,y)$ is not lower than her utility from S, or,
\item For any adversary bids, $n$'s utility from the truthful bid is not lower than her utility from S.
\end{itemize}
\end{lemma}
\begin{proof}
If bidder $n$ submits more than one 2-type bid, all but the top one never enters the winning set. The same holds for all but the top two 1-type bids submitted. For bidder $n$, adding more losing bids (that end up in $L(B)$ regardless of the adversary bids) can only increase its payments. We conclude that a bidder that wishes to maximize utility never submits more than one 2-type bid and two 1-type bids. We thus consider that bidder $n$ prefers to bid $w_1', v_1', v_2'$, some of which may be zero bids or equivalently omitted, over bidding $S$.
\begin{itemize}
\item If $2w_1' \leq v_1' + v_2'$, regardless of the adversary bids, $(2,w_1')$ is never in the winning set. As noted before, in this case including this bid in $S$ could only increase the price for bidder $n$. This is true for any adversary bids and so also in expected utility, and falls into the first case of the lemma. 
\item Otherwise, for any adversary bids, it's impossible for two of bidder $n$'s bids to enter $W(B)$, since only two items can be allocated, and thus only the bids $v_1', v_2'$ are candidates to enter $W(B)$ together. But since $2w_1' > v_1' + v_2'$, and the auctioneer maximizes over social welfare, this can not occur. So bidder $n$ would not bid $v_2'$, as it must be a losing bid. 

Now, given any adversary bid setup $B$, assume bidder $n$ bids $S = \{w_1,v_1\}$. Then either $w_1 \in W(B)$, or $v_1 \in W(B)$, or none of them are in $W(B)$. In the first case, the bidder is better off bidding just $w_1$, as it will still get in the winning set, but without the losing bid $v_1$. In the second case, the bidder is better off bidding just $v_1$. And in the last case the bidder has 0 utility. Since VCG is truthful in dominant strategies for single bids (i.e., bidding truthfully is better than bidding any single bid), and for each case we observe that bidder $n$ is better off bidding some single bid over bidding $\{w_1,v_1\}$, we conclude that for each of the cases the bidder is better off bidding truthfully, and so summing all together, its expected utility bidding truthfully is at least as good as its expected utility bidding $S=\{w_1,v_1\}$. 

\end{itemize}
\end{proof}







\begin{lemma}
\label{lem:vcg_homogenous}
The VCG outcome is homogeneous under multiplication of bids and values, i.e., if we multiply all bids and values by a constant $\alpha$, the winning set remains the same, and payments and utilities are multiplied by $\alpha$ as well.
\end{lemma}
\begin{proof}
Consider adversary bids $(g_1,\theta_1),\ldots, (g_{\tilde{n}},\theta_{\tilde{n}})$, and a bidder of true type $(g,\theta)$ that bids $(g_1',\theta_1'),\ldots, (g_m',\theta_m')$. We show that $$u^{(g,\alpha\theta)}_{(g_1',\alpha\theta_1'),\ldots, (g_m',\alpha\theta_m')}((g_1, \alpha \theta_1), ..., (g_{\tilde{n}}, \alpha \theta_{\tilde{n}})) = \alpha u^{(g,\theta)}_{(g_1',\theta_1'),\ldots, (g_m',\theta_m')}((g_1, \theta_1), ..., (g_{\tilde{n}},  \theta_{\tilde{n}}))$$. First, we notice that the winning set in both cases have the same indices. The price is also determined by the same non-winning set indices. Both the value and the price are multiplied by $\alpha$, and so the overall utility is also multiplied by $\alpha$. Notice that this can easily be extended to apply in the case of expected utilities.
\end{proof}

We use Lemma~\ref{lem:vcg_homogenous} to assume that if we have a bound $M < 1$ for the bid and true type values for both attacker and her adversaries, then we can reset $M=1$ (``normalize''). `

\begin{lemma}
\label{lem:distribution_assumption}
If there exists such a distribution F that an attack $(1,x),(1,y)$ is beneficial for some true type $\theta$, there must exist a distribution $\bar{F}$ with the same attack and true type $\theta$ such that $supp \bar{F} \subseteq [0,\max(x,\theta)]$. 
\end{lemma}
\begin{proof}
If $v_1 > x$, then $y$ never enters the winning set, and so bidding just $(1,x)$ is better, and by VCG truthfulness for single bids, bidding truthfully is even better and the attack cannot be beneficial. If $\tilde{w_2} > x \geq y$ (and following the first part of the proof, we can also assume $v_1 \leq x$), then the attack utility is 0 and since truthfulness in VCG is individually rational, the attack is not beneficial. This holds whether the true type has $g = 1$ or $g=2$. 

Define the pdf $\bar{f}(\psi) = \begin{cases}
\frac{f(\psi)}{F(\max(x, \theta))} & 0\leq \psi \leq \max(x, \theta) \\
0 & \psi > \max(x,\theta)
\end{cases}$
and its corresponding cdf $\bar{F}$. Let $A$ be the event that an adversary with per-item value greater than $\max \{x,\theta\}$ exists. Then

\[
\begin{split}
& E^{q,\bar{F},n}_{\hat{\theta},attack(x,y)}[u] - E^{q,\bar{F},n}_{\hat{\theta},truth}[u] =\\
&(E^{q,\bar{F},n}_{\hat{\theta},attack(x,y)}[u | A] - E^{q,\bar{F},n}_{\hat{\theta},truth}[u | A])\underbrace{Pr_{\bar{F}}[A]}_{=0} + (E^{q,\bar{F},n}_{\hat{\theta},attack(x,y)}[u | \bar{A}] - E^{q,\bar{F},n}_{\hat{\theta},truth}[u | \bar{A}])Pr_{\bar{F}}[\bar{A}] = \\
& E^{q,\bar{F},n}_{\hat{\theta},attack(x,y)}[u | \bar{A}] - E^{q,\bar{F},n}_{\hat{\theta},truth}[u | \bar{A}] = \\
& E^{q,F,n}_{\hat{\theta},attack(x,y)}[u | \bar{A}] - E^{q,F,n}_{\hat{\theta},truth}[u | \bar{A}] \geq \\
& (E^{q,F,n}_{\hat{\theta},attack(x,y)}[u | \bar{A}] - E^{q,F,n}_{\hat{\theta},truth}[u | \bar{A}])Pr_F[\bar{A}] + (\underbrace{E^{q,F,n}_{\hat{\theta},attack(x,y)}[u | A] - E^{q,F,n}_{\hat{\theta},truth}[u | A]}_{\leq 0})Pr_F[A] = \\
& E^{q,F,n}_{\hat{\theta},attack(x,y)}[u] - E^{q,F,n}_{\hat{\theta},truth}[u]
\end{split}
\]
\end{proof}

\begin{lemma}
\label{one_type_assumption_s1}
(1-type). If there exists a true type $(1, \theta)$ that has a beneficial attack, then there also exists a true type $(1, \theta')$ that has a beneficial attack $x,y$ and $y \leq x \leq \theta'$.
\end{lemma}
\begin{proof}
Let a bidder with true type $(1, \theta)$ attack with $x,y$ such that $\theta<x$. We show that the same attack is beneficial for the true type $(1,x)$, for any adversary bids $\hat{\theta_1}, ..., \hat{\theta_{\tilde{n}}}$. That is, for some $(w_1,v_1,v_2)$
\[
\begin{split}
 \tilde{u}_{(1,x), attack(x,y)} - \tilde{u}_{(1,x), truth} \geq  \tilde{u}_{(1,\theta), attack(x,y)} - \tilde{u}_{(1,\theta), truth},
\end{split}
\]

or equivalently 
\begin{equation}
\label{eq:1_type_condition_theta<=x}
\begin{split}
 \tilde{u}_{(1,x), attack(x,y)} - \tilde{u}_{(1,\theta), attack(x,y)} \geq  \tilde{u}_{(1,x), truth} - \tilde{u}_{(1,\theta), truth}.
\end{split}
\end{equation}

First, notice that the L.H.S. of Equation~\ref{eq:1_type_condition_theta<=x} is non-negative, as increasing the value of the item without changing the attack only has the effect of increasing the value in case of winning, without affecting payments.

Now, notice that if the R.H.S. of Equation~\ref{eq:1_type_condition_theta<=x} is strictly positive, then the bid $(1,x)$ is in the winning set when bidding it truthfully. Thus it is also in the winning set when attacking (adding the $(1,y)$ may only add winning events), and the L.H.S. equals $x - \theta$ (the payments remain the same as the attack is the same, but the value the bidder extracts changes by $x - \theta$). We now show by case analysis that the R.H.S. is always smaller than this amount:
\[
\begin{split}
& \tilde{u}_{(1,x), truth} - \tilde{u}_{(1,\theta), truth} = \\
& \begin{cases}
x- \theta & v_2 \leq \theta, w_1 \leq \frac{v_1+v_2}{2} \\
x - v_2 & \theta < v_2 \leq x, w_1 \leq \frac{v_1 + v_2}{2} \\
x - \theta & v_2 \leq \theta, \frac{v_1 + v_2}{2} \leq w_1 \leq \frac{\theta+v_1}{2} \\
x - 2w_1 + v_1 & v_2 \leq \theta, \frac{\theta+v_1}{2} < w_1 \leq \frac{x+v_1}{2} \\
x - 2w_1 + v_1 & \theta < v_2 \leq x, \frac{v_1 + v_2}{2} < w_1 \leq \frac{x+v_1}{2} \\
0 & Otherwise. 
\end{cases}
\end{split}
\]

In any of the cases it is not too hard to show that given the case conditions, the case expression is less than $x - \theta$. 
Notice that this can be extended to apply to the case of expected utilities.
\end{proof}

\section{$q=1$ induces Bayesian Resilience}
The following section constitutes a proof for Lemma~\ref{lem:q=1_analysis} in Section 2.

\qanalysis*

\begin{proof}
By Theorem~\ref{thm:attack_format} we can always assume a 1-type valuation $\theta$ is 1. Then define 
\[
\begin{split}
& \Delta E \stackrel{def}{=} E^{q=1,F,n}_{(g,1),truth}[u] - E^{q=1,F,n}_{(g,1),attack(x,y)}[u] = \\ 
& E^{k=\tilde{n},F,n}_{(g,1),truth}[u] - E^{k=\tilde{n},F,n}_{(g,1),attack(x,y)}[u] =  E_{v_1,v_2}(\tilde{\Delta}(v_1, v_2)),
\end{split}
\] where
$$\tilde{\Delta}(v_1, v_2) \,{=}\, \tilde{u}_{(g,1), truth}(w_1=0,v_1,v_2) - \tilde{u}_{(g,1), attack(x,y)}(w_1=0, v_1, v_2).$$ 
It thus suffices to show $\forall v_1, v_2, \tilde{\Delta}(v_1,v_2) \geq 0$ to attain $\Delta E \geq 0$. For a 1-type bidder ($g=1$), we have
\[
\begin{split}
    \tilde{\Delta}(v_1,v_2) = \begin{cases}
    1 - v_2 \geq 0 & x < v_2 \\
    0 & x \geq v_2 \geq y \\
    y - v_2 > 0 & x, v_1 \geq y > v_2 \\
    2v_1 - v_2 \geq 0 & x\geq y > v_1. \\
    \end{cases}
\end{split}
\]
For a 2-type bidder, 
we can restrict attention to the cases where the two attack bids are the winning set; otherwise the attacker has non-positive utility, while the truthful strategy in VCG has non-negative utility as VCG is individually rational. We thus only consider $y\geq v_1$. However, in this case, for any $v_1, v_2$, we have
$$\tilde{u}_{(2,\theta), truth} = (2\theta - v_1 - v_2) \mathbbm{1}_{\theta\geq v_1} \geq (2\theta - 2v_1)\mathbbm{1}_{\theta\geq v_1} \geq 2\theta - 2v_1 = 
\tilde{u}_{(2,\theta), attack(1,y)}.$$

Note that it is actually somewhat redundant to prove truthful bidding for the 2-type attacker, Since given $q=1$ it is not in the support of the distribution to ever have a 2-type attacker. 
\end{proof}

\section{Proof of Lemma~\ref{two_type_q*}}
The following section constitutes a proof for Lemma~\ref{two_type_q*} in Section 3.

\twotypeq*

\begin{proof}

In this proof, we use the simplified notations 
$u^{truth}_{(g_{adv},\theta_{adv})} = u^{\hat{\theta}=(2,\theta)}_{\theta'_1=(1,1)}((g_{adv},\theta_{adv}))$ for the utility of bidder $2$ when bidding truthfully faced with a truthful adversary of type $(g_{adv},\theta_{adv})$. We write $u^{attack(x,y)}_{(g_{adv},\theta_{adv})} = u^{\hat{\theta}=(2,\theta)}_{\theta'_1=(1,x), \theta'_2=(1,y) }((g_{adv},\theta_{adv}))$ for the utility of bidder $2$ when bidding the attack $(1,x),(1,y)$ faced with a truthful adversary of type $(g_{adv},\theta_{adv})$. Notice that the difference from the notations used in the proof of Lemma~\ref{lem:one_type_q*} is that now bidder $2$ is a 2-type with per-item value of $\theta$. 

A similar argument to Lemma~\ref{lem:one_type_q*} yields that it suffices to show that for any parameters $\theta, \theta_{adv}, x, y$, $$u^{truth}_{(2,\theta_{adv})} - u^{attack(x,y)}_{(2,\theta_{adv})} + 2(u^{truth}_{(1,\theta_{adv})} - u^{attack(x,y)}_{(1,\theta_{adv})})$$

We first derive the explicit expression for the utilities:
\[
\begin{aligned}
& u^{attack(x,y)}_{(2,\theta_{adv})} = 
      \begin{cases}
2\theta & 0 \leq \theta_{adv} \leq \frac{y}{2} \\
2\theta - 2\theta_{adv} + y & \frac{y}{2} < \theta_{adv} \leq \frac{x}{2} \\
2\theta - 4\theta_{adv} + y + x & \frac{x}{2} < \theta_{adv} \leq \frac{x + y}{2}  \\
0 & otherwise,
\end{cases} \\[6mm]
    & u^{truth}_{(2,\theta_{adv})} = 
     \begin{cases}
2\theta - 2\theta_{adv} & 0 \leq \theta_{adv} \leq \theta \\
0 & otherwise,
\end{cases}
\end{aligned} \hspace{0.5cm} 
\begin{aligned}
& u^{attack(x,y)}_{(1,\theta_{adv})} = 
 \begin{cases}
    2\theta - 2\theta_{adv} & 0\leq \theta_{adv} \leq y \\
    -y & y < \theta_{adv} \leq 1, \\
    \end{cases} \\[16mm]
& u^{truth}_{(1,\theta_{adv})} = \begin{cases}
2\theta - \theta_{adv} & 0 \leq \theta_{adv} \leq 2\theta \\
0 & otherwise. \\
\end{cases}.
\end{aligned} \hspace{1cm}
\]

Combining the utility functions we have (for $0\leq \theta_{adv}\leq 1$)
\[ 
\begin{split}
 & u^{truth}_{(2,\theta_{adv})} - u^{attack(x,y)}_{(2,\theta_{adv})} + 2(u^{truth}_{(1,\theta_{adv})} - u^{attack(x,y)}_{(1,\theta_{adv})})  = \\
& \begin{cases}
6\theta - 4\theta_{adv} + 2y - u^{attack(x,y)}_{(2,\theta_{adv})}  & y\leq \theta_{adv} \leq \theta \\
4\theta - 2\theta_{adv} + 2y - u^{attack(x,y)}_{(2,\theta_{adv})}  & \max\{y,\theta\} \leq \theta_{adv} \leq 2\theta \\
y - u^{attack(x,y)}_{(2,\theta_{adv})} & \max\{y,2\theta\} \leq  \theta_{adv} \\
2\theta - u^{attack(x,y)}_{(2,\theta_{adv})}  & \theta_{adv} \leq \min\{y, \theta\}\\
 2\theta_{adv} - u^{attack(x,y)}_{(2,\theta_{adv})}  & \theta \leq \theta_{adv} \leq y\\
4\theta_{adv} - 4\theta - u^{attack(x,y)}_{(2,\theta_{adv})}  & 2\theta \leq \theta_{adv} \leq y . \\
\end{cases}
\end{split}
\]

In essence, we enumerate over all 6 possible positions of $\theta_{adv}$  relative to $\{y, \theta, 2\theta \}$. 

We analyze each case to be non-negative, using the explicit utility expressions. The first case expression is non-negative since $u^{attack(x,y)}_{(2,\theta_{adv})} \leq 2\theta$ for whatever range $\theta_{adv}$ is in, and also $\theta \geq \theta_{adv}, y\geq 0$. 
The second case expression is non-negative since for this case we have $\theta_{adv}\geq y \geq \frac{y}{2}$, so  $u^{attack(x,y)}_{(2,\theta_{adv})} \leq 2\theta - 2\theta_{adv} + y$, and so the expression is at least $2\theta + y \geq 0$. The third case expression is non-negative since for this case we also have $u^{attack(x,y)}_{(2,\theta_{adv})} \leq 2\theta - 2\theta_{adv} + y$, and so the expression is at least $2\theta_{adv} - 2\theta$, but $\theta_{adv} \geq 2\theta \geq \theta$. The fourth case expression is non-negative since it always holds that $u^{attack(x,y)}_{(2,\theta_{adv})} \leq 2\theta$. The fifth case expression is non-negative since for this case, $u^{attack(x,y)}_{(2,\theta_{adv})} \leq 2\theta \leq 2\theta_{adv}$. The sixth case expression is non-negative since for this case,  $u^{attack(x,y)}_{(2,\theta_{adv})} \leq 2\theta \leq 4\theta \leq 4\theta_{adv} - 4\theta$.
\end{proof}

\section{Full Description of $\tilde{u}_{\hat{\theta},S}(w_1, v_1, v_2)$}
\label{app:utilities_full_form}
The following section gives the full form of $\tilde{u}_{\hat{\theta},S}(w_1, v_1, v_2)$ that is used in the reduction to a polynomial positivity decision problem in Section~\ref{sec:beta_distributions}. Recall that we claim that $\tilde{u}_{\hat{\theta},S}$ is generally of the form of polynomials in $w_1,v_1,v_2,x,y,\theta$ multiplied by an indicator function with bounds which are also polynomial in these parameters. This can be verified directly from the explicit form given below.

By Theorem~\ref{thm:attack_format} we can assume for 1-type bidders that $\theta = 1$. 
We have

\[
\begin{aligned}
&
\tilde{u}_{(1,1),truth}(w_1,v_1, v_2) = \\
    & \begin{cases}
    1 - v_2 & 0\leq w_1 \leq \frac{v_1 + v_2}{2} \\
    1 - 2w_1 + v_1 & \frac{v_1 + v_2}{2} < w_1 \leq \frac{v_1 + 1}{2} \\
    0 & otherwise.
    \end{cases} \\[2.5cm]
\end{aligned} \hspace{0.5cm} 
\begin{aligned}
& \tilde{u}_{(1,1),attack(x,y)}(w_1, v_1, v_2) = \\
    & \begin{cases}
    1 - 2v_1 & v_1 \leq y, 0 \leq w_1 \leq \frac{y + v_1}{2} \\
    1 - 2w_1 + y - v_1 & 0 \leq v_2 \leq v_1 \leq y, \frac{y + v_1}{2} < w_1 \leq \frac{x + v_1}{2} \\
    1 - 4w_1 + y + 1 & 0 \leq v_2 \leq v_1 \leq y, \frac{x + v_1}{2} < w_1 \leq \frac{x + y}{2} \\
    1-y & v_2 < y < v_1 \leq 1, 0 \leq w_1 \leq  \frac{y + v_1}{2} \\
    1-v_2 & y < v_2 \leq x, v_1 \leq 1, 0 \leq w_1 \leq \frac{v_1 + v_2}{2} \\
    1-2w_1 + v_1 & y < v_2 \leq x,v_1 \leq 1, \frac{v_1 + v_2}{2} < w_1 \leq \frac{v_1 + x}{2} \\
    0 & otherwise.
    \end{cases} 
\end{aligned} \hspace{1cm}
\]
    
For 2-type bidders, we have
\[
\begin{aligned}
   & \tilde{u}_{(2,\theta),truth}(w_1, v_1, v_2) = \\
   & \begin{cases}
    2\theta - v_1 - v_2 & v_1 + v_2 \leq 2\theta, w_1 \leq \frac{v_1 + v_2}{2} \\
    2\theta - 2w_1 & v_1 + v_2 \leq 2\theta, \frac{v_1 + v_2}{2} < w_1 \leq \theta \\
    0 & otherwise.
    \end{cases} \\[2.5cm]
\end{aligned} \hspace{0.5cm} 
\begin{aligned}
    & \tilde{u}_{(2,\theta),attack(x,y)}(w_1, v_1, v_2) = \\
    & \begin{cases}
    2\theta - 2v_1 & v_1 \leq y, w_1 \leq \frac{y + v_1}{2} \\
    2\theta - 2w_1 + y - v_1 & v_1 \leq y, \frac{y + v_1}{2} < w_1 \leq \frac{x + v_1}{2} \\
    2\theta - 4w_1 + y + x & v_1 \leq y, \frac{x + v_1}{2} < w_1 \leq \frac{x + y}{2} \\
    -y & v_2 \leq y < v_1,  w_1 \leq  \frac{y + v_1}{2} \\
    -v_2 & y < v_2 \leq x , w_1 \leq \frac{v_1 + v_2}{2} \\
    -2w_1 + v_1 & y < v_2 \leq x, \frac{v_1 + v_2}{2} < w_1 \leq \frac{v_1 + x}{2} \\
    0 & otherwise.
    \end{cases}
\end{aligned} \hspace{1cm}
\]
    
\section{Proof that for the uniform distribution 1-type, $q^*\leq \frac{1}{2}$}
\label{app:uniform_half_1type}

This appendix constitutes the proof for the first part of Theorem~\ref{thm:uniform_half}. By Theorem~\ref{thm:attack_format}, it suffices to show that for any $q > \frac{1}{2}, n\geq 2$, and attack parameters $x,y$ (we omit $F=UNI([0,1])$ from the notations as it is the scope of discussion for this entire appendix),

$$E^{q,n}_{\hat{\theta}=(1,1),truth} - E^{q,n}_{\hat{\theta}=(1,1),attack(x,y)} =  \sum_{k=0}^{\tilde{n}} \binom{\tilde{n}}{k} q^k (1-q)^{\tilde{n}-k} (E^{k,n}_{\hat{\theta}=(1,1),truth} - E^{k,n}_{\hat{\theta}=(1,1),attack(x,y)}) \geq 0$$

We prove the claims
\begin{claim}
\label{clm:uniform_expected_positive_general_k}
For any $k\geq 1$, and parameters $n,x,y$, $E^{k,n}_{\hat{\theta}=(1,1),truth} - E^{k,n}_{\hat{\theta}=(1,1),attack(x,y)} \geq 0$
\end{claim}

\begin{claim} 
\label{clm:uniform_expected_positive_edge_case}
$$\sum_{k\in \{0,1,\tilde{n}\}} \binom{\tilde{n}}{k} (E^{k,n}_{\hat{\theta}=(1,1),truth} - E^{k,n}_{\hat{\theta}=(1,1),attack(x,y)})\geq 0$$
\end{claim}

Together, these claims establish that indeed
\[
\begin{split}
    & \sum_{k=0}^{\tilde{n}} \binom{\tilde{n}}{k} q^k (1-q)^{\tilde{n}-k} (E^{k,n}_{\hat{\theta}=(1,1),truth} - E^{k,n}_{\hat{\theta}=(1,1),attack(x,y)}) \stackrel{Claim~\ref{clm:uniform_expected_positive_general_k}}{\geq} \\
& \sum_{k\in \{0,1,\tilde{n}\}} \binom{\tilde{n}}{k} q^k (1-q)^{\tilde{n}-k} (E^{k,n}_{\hat{\theta}=(1,1),truth} - E^{k,n}_{\hat{\theta}=(1,1),attack(x,y)}) \stackrel{Claim~\ref{clm:uniform_expected_positive_general_k} \text{ and } q\geq 1-q}{\geq}  \\
& (1-q)^{\tilde{n}} \sum_{k\in \{0,1,\tilde{n}\}} \binom{\tilde{n}}{k} (E^{k,n}_{\hat{\theta}=(1,1),truth} - E^{k,n}_{\hat{\theta}=(1,1),attack(x,y)})
\stackrel{Claim~\ref{clm:uniform_expected_positive_edge_case}}{\geq} 0 , 
\end{split}
\]

We now move to prove the claims. 

\begin{proof} \textbf{(Claim.~\ref{clm:uniform_expected_positive_general_k})}
 Consider some $n\geq 2$ and $k\geq 1$. The utility expressions as given in Appendix D for 1-types satisfy $\tilde{u}_{(1,1),truth}(w_1, v_1, v_2) \geq \tilde{u}_{(1,1),attack(x,y)}(w_1, v_1, v_2)$, except perhaps the domain $v_1\leq y, \frac{y+v_1}{2} \leq w_1 \leq \frac{x+y}{2}$. That is since if $w_1 \leq \frac{v_1 + v_2}{2}$, then $\tilde{u}_{(1,1),truth}(w_1, v_1, v_2) = 1 - v_2 \geq \tilde{u}_{(1,1),attack(x,y)}(w_1, v_1, v_2)$ in any of the possible domains. If $\frac{1 + v_1}{2} \leq w_1$, we have $\tilde{u}_{truth}(w_1, v_1, v_2) = \tilde{u}_{attack(x,y)}(w_1, v_1, v_2) = 0$. If $\frac{v_1 + v_2}{2} \leq w_1 \leq \frac{v_1 + 1}{2}$, then $\tilde{u}_{truth}(w_1, v_1, v_2) = 1 - 2w_1 + v_1$. For $w_1 \leq \frac{y + v_1}{2}$, this expression is at least $1-y \geq 0$ (4th and 7th cases of the attack expression). The condition on $w_1$ can not coincide with the 1st and 5th case of the attack expression, and $\tilde{u}_{truth}(w_1, v_1, v_2) = \tilde{u}_{attack(x,y)}(w_1, v_1, v_2)$ for the 6th case of the attack expression. For the remaining 2nd and 3rd cases of the attack expression, we have $\tilde{u}_{truth}(w_1, v_1, v_2) - \tilde{u}_{attack(x,y)}(w_1, v_1, v_2) \geq (1 - 2w_1 + v_1) - (1 - 2w_1 + y - v_1) = 2v_1 - y$. As this is the only possible negative utility expression, to prove the claim over the expected utilities, it is enough to prove it over the conditional expected utilities over this domain. 
Thus,
\[
\begin{split}
    & E^{k,n}_{\hat{\theta}=(1,1),truth} - E^{k,n}_{\hat{\theta}=(1,1),attack(x,y)} \geq \\
    & 2^{\tilde{n}-k} \int_{v_1=0}^{y} v_1^{k-1} \int_{w_1=\frac{y+v_1}{2}}^{\frac{x+y}{2}} w_1^{\tilde{n}-k-1}(2v_1 - y) dw_1 dv_1 = \\
    & 2^{\tilde{n}-k}\int_{v_1=0}^{y} v_1^{k-1}(2v_1 - y)((\frac{x+y}{2})^{\tilde{n}-k}-(\frac{y+v_1}{2})^{\tilde{n}-k})dv_1 = \\
    & (\frac{2}{k+1} - \frac{1}{k})y^{k+1}(x+y)^{\tilde{n}-k} + \sum_{i=0}^{\tilde{n}-k} \binom{\tilde{n}-k}{i}(\frac{1}{\tilde{n}-i} - \frac{2}{\tilde{n}-i+1})y^{\tilde{n}+1} \geq \\
    & \sum_{i=0}^{\tilde{n}-k} \binom{\tilde{n}-k}{i}(\frac{2}{k+1} - \frac{1}{k} + \frac{1}{\tilde{n}-i} - \frac{2}{\tilde{n}-i+1}) y^{\tilde{n}+1}
\end{split}
\]

For each summand, it holds that
\[
\begin{split}
    & \frac{2}{k+1} - \frac{1}{k} + \frac{1}{\tilde{n}-i} - \frac{2}{\tilde{n}-i+1} = \\
    & \frac{k-1}{k(k+1)} - \frac{\tilde{n}-i-1}{(\tilde{n}-i)(\tilde{n}-i+1)} \stackrel{k\leq \tilde{n}-i}{\geq} \\
    &  \frac{k-1}{k(k+1)} - \frac{k-1}{k(k+1)} = 0,
\end{split}
\]
where the inequality holds since the expression $f(x) = \frac{x-1}{x(x+1)}$ is weakly monotonically decreasing in $x$ for $x\geq 3$ (and so correspondingly in the integer series), and $f(2) = f(3) = \frac{1}{6}$. 

\end{proof}

\begin{proof} \textbf{(Claim.~\ref{clm:uniform_expected_positive_edge_case})}

We multiply the expression of the statement by the non-negative factor $(x+y)(\tilde{n}+1)2^{\tilde{n}}$ for a nicer representation of the direct calculation result. We then have 

\begin{equation}
\label{eq:expectation_polynomial_bound}
\begin{split}
    & (x+y)({\tilde{n}}+1)2^{\tilde{n}} \sum_{k\in \{0,1,\tilde{n}\}} \binom{\tilde{n}}{k} (E^{k,n}_{\hat{\theta}=(1,1),truth} - E^{k,n}_{\hat{\theta}=(1,1),attack(x,y)}) = \\
    &  (x+y) \overbrace{\left(-2 x^{{\tilde{n}}+1}-2 {\tilde{n}} (x+1)^{\tilde{n}}+2 {\tilde{n}} x (x+1)^{\tilde{n}}-4 (x+1)^{\tilde{n}}-2 y^{{\tilde{n}}+1}+2 (y+1)^{{\tilde{n}}+1}+3 \cdot  2^{{\tilde{n}}+1}-3\right)}^{exp1} \\
    & -2 (x+y)^{\tilde{n}} \overbrace{\left(({\tilde{n}}-1) x^2-x \left(({\tilde{n}}-1)^2 y+{\tilde{n}}+1\right)+y ({\tilde{n}} ({\tilde{n}} (-y)+{\tilde{n}}+y)-1)\right)}^{exp2}
    \end{split}
    \end{equation}

First we note that for $exp1$, we can make the substitution $\frac{2d}{{\tilde{n}}} = 1 - x$ (which also yields $(2 - \frac{2d}{{\tilde{n}}}) = 1+x$), and have

\begin{equation}
\label{eq:bound_exp1}
\begin{split}
    & -2 x^{{\tilde{n}}+1}-2 {\tilde{n}} (x+1)^{\tilde{n}}+2 {\tilde{n}} x (x+1)^{\tilde{n}}-4 (x+1)^{\tilde{n}}-2 y^{{\tilde{n}}+1}+2 (y+1)^{{\tilde{n}}+1}+3 \cdot 2^{{\tilde{n}}+1}-3 \stackrel{\left(-2x^{{\tilde{n}}+1} - 2y^{{\tilde{n}}+1} \geq -4\right)}{\geq} \\
    & -2 {\tilde{n}} (x+1)^{\tilde{n}}+2 {\tilde{n}} x (x+1)^{\tilde{n}}-4 (x+1)^{\tilde{n}}+3 \cdot 2^{{\tilde{n}}+1}-7 \geq \qquad \qquad \qquad \qquad \left( -4(x+1)^{\tilde{n}} + 3\cdot  2^{{\tilde{n}}+1} \leq 2^{{\tilde{n}}+1} \right) \\
    & 2^{{\tilde{n}}+1}-7 -2{\tilde{n}}(1-x)(x+1)^{\tilde{n}} = \\
    & 2^{{\tilde{n}}+1}-7 -4d(2 - \frac{2d}{{\tilde{n}}})^{\tilde{n}} = \\
    & 2^{{\tilde{n}}+1}-7 -2^{\tilde{n}} \cdot 4d (1 - \frac{d}{{\tilde{n}}})^{\tilde{n}} \geq \qquad \qquad \qquad \qquad \qquad \qquad \qquad \qquad \qquad \qquad \left( (1 - \frac{d}{{\tilde{n}}})^{\tilde{n}} \leq e^{-d} \right) \\
    & 2^{{\tilde{n}}+1}-7 -2^{\tilde{n}} \cdot 4d e^{-d} \geq \qquad \qquad \qquad \qquad \qquad \qquad \qquad \qquad \qquad \qquad \qquad \left( \max_{d\geq 0} \{ de^{-d} \} = \frac{1}{e} \right)  \\
    &  2^{{\tilde{n}}+1}-7 -2^{\tilde{n}}\frac{4}{e} = (4 - \frac{8}{e})2^{{\tilde{n}}-1} - 7 \geq \\
    & 2^{{\tilde{n}}-1} - 7
\end{split}
\end{equation}

For $exp2$, notice it is preceded with a minus sign, so we aim to show it is a negative expression. If $x+y \geq  (\frac{{\tilde{n}}}{{\tilde{n}}-1})^2$, we have

\begin{equation}
\label{eq:bound_exp2}
\begin{split}
& ({\tilde{n}}-1) x^2-x \left(({\tilde{n}}-1)^2 y+{\tilde{n}}+1\right)+y ({\tilde{n}} ({\tilde{n}} (-y)+{\tilde{n}}+y)-1) \leq \qquad \qquad \left( ({\tilde{n}}-1)x^2 -({\tilde{n}}+1)x\leq 0 \right) \\
& -x ({\tilde{n}}-1)^2 y +y ({\tilde{n}} ({\tilde{n}} (-y)+{\tilde{n}}+y)-1) \leq \qquad \qquad \qquad \qquad \qquad \left( -{\tilde{n}}^2y^2 + {\tilde{n}}y^2 -y{\tilde{n}} \leq -({\tilde{n}}-1)^2y^2 \right) \\
& {\tilde{n}}^2 y-({\tilde{n}}-1)^2 y (x+y) \leq \qquad \qquad \qquad \qquad \qquad \qquad \qquad \qquad \qquad \qquad \left( x+y \geq  (\frac{{\tilde{n}}}{{\tilde{n}}-1})^2 \right) \\
& {\tilde{n}}^2y - {\tilde{n}}^2y = 0.
\end{split}
\end{equation}

In this case, both expressions (adjusting for the minus sign) are non-negative for ${\tilde{n}}\geq 4$, and the result follows.

Regardless of the condition on $x+y$, the development of Equation~\ref{eq:bound_exp2} (where we continue without applying the step $x+y \geq (\frac{{\tilde{n}}}{{\tilde{n}}-1})^2$) yields 

\begin{equation}
\label{eq:general_bound_exp2}
    ({\tilde{n}}-1) x^2-x \left(({\tilde{n}}-1)^2 y+{\tilde{n}}+1\right)+y ({\tilde{n}} ({\tilde{n}} (-y)+{\tilde{n}}+y)-1) \leq {\tilde{n}}^2y
\end{equation}

If $x+y \leq (\frac{{\tilde{n}}}{{\tilde{n}}-1})^2$ (and excluding the trivial case of $x=y=0$), we can divide the R.H.S. of Equation~\ref{eq:expectation_polynomial_bound} by a factor of $x+y$ and have
\[
\begin{split}
    &  \left(-2 x^{{\tilde{n}}+1}-2 {\tilde{n}} (x+1)^{\tilde{n}}+2 {\tilde{n}} x (x+1)^{\tilde{n}}-4 (x+1)^{\tilde{n}}-2 y^{{\tilde{n}}+1}+2 (y+1)^{{\tilde{n}}+1}+3 \cdot 2^{{\tilde{n}}+1}-3\right) \\
    & -2 (x+y)^{{\tilde{n}}-1} \left(({\tilde{n}}-1) x^2-x \left(({\tilde{n}}-1)^2 y+{\tilde{n}}+1\right)+y ({\tilde{n}} ({\tilde{n}} (-y)+{\tilde{n}}+y)-1)\right) \stackrel{Eq.~\ref{eq:bound_exp1}}{\geq} \\
    & 2^{{\tilde{n}}-1} - 7 - 2 (x+y)^{{\tilde{n}}-1} \left(({\tilde{n}}-1) x^2-x \left(({\tilde{n}}-1)^2 y+{\tilde{n}}+1\right)+y ({\tilde{n}} ({\tilde{n}} (-y)+{\tilde{n}}+y)-1)\right) \stackrel{Eq.~\ref{eq:general_bound_exp2}}{\geq} \\
    & 2^{{\tilde{n}}-1} - 7 - 2 (x+y)^{{\tilde{n}}-1} {\tilde{n}}^2 y \geq \qquad \qquad \qquad \qquad \qquad \qquad \qquad \qquad \qquad \qquad \left(x+y \leq (\frac{{\tilde{n}}}{{\tilde{n}}-1})^2 \right) \\
    & 2^{{\tilde{n}}-1} - 7 - 2(1 + \frac{1}{{\tilde{n}}-1})^{2({\tilde{n}}-1)} {\tilde{n}}^2y \geq \qquad \qquad \qquad \qquad \qquad \qquad \qquad \qquad \left(1 + \frac{1}{{\tilde{n}}-1})^{2({\tilde{n}}-1)} \leq e^2 \right) \\
    & 2^{{\tilde{n}}-1} - 7 - 2e^2{\tilde{n}}^2,
    \end{split}
    \]
    
which is non-negative for $\tilde{n} \geq 13$. For $\tilde{n} < 13$ we verify directly (using Partial CAD) that the expression is always non-negative. 
\end{proof}
    
\section{Proof that for split attacks with uniform distribution, $q^* \leq \frac{1}{2}$} 
\label{app:UniformDistProof}

This appendix constitutes the proof for the second part of Theorem~\ref{thm:uniform_half}.
The proofs in this section focus on showing that for the split attack where $\theta = x = y = 1$, $q^* = \frac{1}{2}$. This bound generalizes to any form of split attack with $\theta = x = y$, by the lemma below. That is since for split attacks $\theta = x=y$, and it is w.l.o.g. to assume at least one of these parameters is 1, it follows that it suffices to analyze the $q^*$ value where $\theta = x = y = 1$. Note that the uniform distribution falls into the settings of the lemma with $\alpha = \beta = 1$.
\begin{lemma}
\label{lem:beta=1_assumptions}
For any Beta distribution $F_{\alpha,\beta}$ with $\beta=1$, it is w.l.o.g. to assume either $x=1$ or $\theta=1$, for the analysis of 2-types $q^*$. 
\end{lemma}
\begin{proof}
Beta distributions $F_{\alpha,\beta}$ with $\beta=1$ have the general form of $F(x) = x^{\alpha}$, or $f(x) = \alpha x^{\alpha - 1}$. We know by Lemma~\ref{lem:distribution_assumption} that there is some distribution $\tilde{F}_{\alpha, \beta}$ with $supp \tilde{F} \subseteq [0, \max \{x,\theta\}]$ such that $q^*$ is at least as in the original distribution, and by Lemma~\ref{lem:vcg_homogenous} we know we can rescale $\tilde{F}$ to have support in $[0,1]$ and either $x=1$ or $\theta = 1$. In general, we are not guaranteed that the distribution after these two transformations is the same as the original distribution, and so we generally use these two lemmas to study $q^*$ values that allow arbitrary distributions. We now show that in the case of distributions with $F(x) = x^{\alpha}$, we do end up with the same distribution. Let $m = \max \{x, \theta\}$. Then 

\[
\begin{split}
& \tilde{f}(\psi) = \begin{cases}
   \frac{\alpha \psi^{\alpha - 1}}{\int\limits_{0}^m \alpha \psi^{\alpha - 1}d\psi} & \psi \leq m \\
   0 & \psi > m
\end{cases} = \begin{cases}
   \frac{\alpha \psi^{\alpha - 1}}{m^{\alpha}} & \psi \leq m \\
   0 & \psi > m
\end{cases} \implies \tilde{F}(\psi) = \begin{cases}
\frac{\psi^{\alpha}}{m^{\alpha}} & \psi \leq m \\
1 & \psi > m \\
\end{cases}
\end{split}
\]

After rescaling, we have
\[
\begin{split}
   \tilde{F}_{rescale}(\psi) = \tilde{F}(m\psi) = \frac{(m\psi)^{\alpha}}{m^{\alpha}} = \psi^{\alpha} = F(\psi),
\end{split}
\]

and so we conclude that any split attack for $F_{\alpha, \beta}$ with $\beta = 1$ has at most the $q^*$ as the split attack $\theta = x = y = 1$ for that same distribution. 
\end{proof}

The structure of the rest of the section is as follows: In Lemma~\ref{lem:split_is_monotone_k} we show that for the split attack, $$E^{k,F,n}_{(2,1),truth}[u] - E^{k,F,n}_{(2,1),attack(1,1)}[u]$$ is monotone in $k$ (the more 1-type adversaries a bidder faces, the more beneficial it is to be truthful). In Lemma~\ref{lem:monotone_condition} we show that this type of monotonicity yields a ``clear cut'' $q^*$ value: Any $q < q^*$ does not have truthful bidding as a BNE, while every $q > q^*$ does, and this $q^*$ value has $E^{q,F,n}_{(2,1),truth}[u] = E^{q,F,n}_{(2,1),attack(1,1)}[u]$. We conclude by  showing in Lemma~\ref{lem:uniform_split_half} that $q^* = \frac{1}{2}$ satisfies this equality for the split attack with uniform distribution.

\begin{lemma}
\label{lem:split_is_monotone_k}
For any distribution $F$, number of bidders $n$ and number of 1-type adversary bidders $0\leq k \leq \tilde{n}-1$ \begin{equation}
\label{eq:split_is_monotone_k}
    E^{k+1,F,n}_{(2,1),truth}[u] - E^{k+1,F,n}_{(2,1),attack(1,1)}[u] \geq E^{k,F,n}_{(2,1),truth}[u] - E^{k,F,n}_{(2,1),attack(1,1)}[u]
    \end{equation}
    
\end{lemma}

\begin{proof}
We rewrite the expectation of the L.H.S. as an expectation over two independent choices: The first is where we sample adversary per-item values $\theta_1',...,\theta_{\tilde{n}}'$, and the second where we uniformly choose $k+1$ of the adversaries to be 1-type. For this purpose, let $[\tilde{n}] = \{1,...,\tilde{n}\}$, and let $\mathcal{I}_k = \bigcup \{K\subseteq [\tilde{n}] \text{ s.t. } |K|=k \}$ be the set of all index sets of size $k$. Given some choice $I\in \mathcal{I}_k$, let $t_I$ be the function that maps per-item values $\theta_1',...,\theta_{\tilde{n}}'$ into the ordered list $(1,v_1),...,(1,v_k), (2,w_1), ..., (2,w_{\tilde{n}-k-1}))$, where a per-item value $\theta_i'$ is assigned to be a 1-type if and only if $i\in I$. Let $\delta u(v_1,...,v_k, w_1, ..., w_{\tilde{n}-k-1}) = u^{(2,1)}_{truth}(v_1,...,v_k, w_1, ..., w_{\tilde{n}-k-1}) - u^{(2,1)}_{attack(1,1)}(v_1,...,v_k, w_1, ..., w_{\tilde{n}-k-1})$. Then 
\[
\begin{split}
    & E^{k+1,F,n}_{(2,1),truth}[u] - E^{k+1,F,n}_{(2,1),attack(1,1)}[u] = E_{\theta_1',...,\theta_{\tilde{n}}' \sim F}[E_{I \sim UNI(\mathcal{I}_{k+1})}[\delta u(t_I(\theta_1',...,\theta_{\tilde{n}})]]
\end{split}
\]
For the R.H.S. side we use a similar rewrite, but observe that choosing $k$ adversaries to be 1-type uniformly ($I \sim UNI(\mathcal{I}_k)$) is the same as choosing $k+1$ adversaries to be 1-type uniformly and then uniformly choosing one of the 1-types to ``flip'' to a 2-type (Choose $I\setminus \{i\}$ where $I \sim UNI(\mathcal{I}_{k+1}), i \sim UNI(I)$). Thus
\[
\begin{split}
    & E^{k,F,n}_{(2,1),truth}[u] - E^{k,F,n}_{(2,1),attack(1,1)}[u] = \\
    & E_{\theta_1',...,\theta_{\tilde{n}}' \sim F}[E_{I \sim UNI(\mathcal{I}_k)}[\delta u(t_I(\theta_1',...,\theta_{\tilde{n}}))]] = \\
    & E_{\theta_1',...,\theta_{\tilde{n}}' \sim F}[E_{I\sim UNI(\mathcal{I}_{k+1})}[E_{i\sim UNI(I)}[\delta u(t_{I\setminus \{i\}}(\theta_1',...,\theta_{\tilde{n}}))]]]
\end{split}
\]
So, to prove \eqref{eq:split_is_monotone_k}, it suffices to show that for any $I \in \mathcal{I}_{k+1}$ and $\theta_1',...,\theta'_{\tilde{n}}$,
$$\delta u(t_I(\theta_1',...,\theta'_{\tilde{n}}) \geq E_{i\sim UNI(I)}[\delta u(t_{I\setminus \{i\}}(\theta_1',...,\theta'_{\tilde{n}}))].$$

If the sampled $i\in I$ is the index corresponding to any of $v_3,\ldots,v_{k+1}$, then
$\delta u(t_{I}(\theta_1',...,\theta'_{\tilde{n}})) = \delta u(t_{I\setminus \{i\}}(\theta_1',...,\theta'_{\tilde{n}}))$, since ``flipping'' any of $v_3,...,v_{k+1}$ to be 2-type does not effect the utilities (it does not change the top two 1-types and either does not change the top 2-type, or if it does, then $v_1 + v_2 \geq 2w_1$ and the top 2-type does not effect either the truthful or attack utility). We introduce the notation
$$ \delta u(v_1,...,v_k,v_{k+1},w_1,...,w_{\tilde{n}-k-1})_{v_j \rightarrow 2-type} = \delta u(t_{I\setminus i_j}(\theta_1',...,\theta'_{\tilde{n}})) ,$$
where $i_j$ is the index of $v_j$. We are thus interested in what happens when we apply $v_1 \rightarrow 2-type, v_2 \rightarrow 2-type$, i.e., we need to show
\[
\begin{split}
& \delta u(v_1,...,v_k,v_{k+1},w_1,...,w_{\tilde{n}-k-1}) \\
& \geq E_{i\sim UNI(I)}[\delta u(v_1,...,v_k,w_1,...,w_{\tilde{n}-k})] = \\
& \frac{k-1}{k+1}\delta u(v_1,...,v_k,v_{k+1},w_1,...,w_{\tilde{n}-k-1}) + \\
& \frac{1}{k+1}\delta u(v_1,...,v_k,v_{k+1},w_1,...,w_{\tilde{n}-k-1})_{v_2 \rightarrow 2-type} + \frac{1}{k+1}\delta u(v_1,...,v_k,v_{k+1},w_1,...,w_{\tilde{n}-k-1})_{v_1 \rightarrow 2-type} 
\end{split}
\]
which is equivalent to 
\begin{equation}
\label{eq:flippingInequality_u}
\begin{split}
    & 2\delta u(v_1,...,v_k,v_{k+1},w_1,...,w_{\tilde{n}-k-1}) \geq \\
    & \delta u(v_1,...,v_k,v_{k+1},w_1,...,w_{\tilde{n}-k-1})_{v_2 \rightarrow 2-type} +  \delta u(v_1,...,v_k,v_{k+1},w_1,...,w_{\tilde{n}-k-1})_{v_1 \rightarrow 2-type}
\end{split}
\end{equation}
For the split attack, the utility expressions of Appendix D (with $\theta = x = y = 1$) translate into 
         \[
    \begin{split}
   & \tilde{u}^{(2,1)}_{truth}(w_1, v_1, v_2) =  \begin{cases}
   2 - v_1 - v_2 & w_1 \leq \frac{v_1+v_2}{2} \\
    2 - 2w_1 &  \frac{v_1 + v_2}{2} < w_1   
    \end{cases} \\
    & \tilde{u}^{(2,1)}_{attack(1,1)}(w_1, v_1, v_2) = \begin{cases}
    2 - 2v_1 & w_1 \leq \frac{1 + v_1}{2} \\
    4 - 4w_1 &  \frac{1 + v_1}{2} < w_1 \\
    \end{cases}
     \end{split}
     \]
     
     which yields (letting $w_1' = \max \{v_2, w_1\}$)
         \[
    \begin{split}
   & 2\delta u(v_1,...,v_k,v_{k+1},w_1,...,w_{\tilde{n}-k-1}) \\
   & - \delta u(v_1,...,v_k,v_{k+1},w_1,...,w_{\tilde{n}-k-1})_{v_2 \rightarrow 2-type} - \delta u(v_1,...,v_k,v_{k+1},w_1,...,w_{\tilde{n}-k-1})_{v_1 \rightarrow 2-type} = \\
   & \begin{cases}
       0 &  \max\{v_1, \frac{1+v_1}{2}\} \leq w_1 \\
    2 + 2v_1 - 4w_1 &  v_1 \leq w_1 \leq \frac{1+v_1}{2} \\
    2v_1 - 2w_1 - \delta u(v_1,...,v_{k+1}, w_1, ..., w_{\tilde{n}-k-1})_{v_1 \rightarrow 2-type} & \frac{v_1+v_2}{2} \leq w_1 \leq v_1 \\
    2v_1 - 2w_1' - \delta u(v_1,...,v_{k+1}, w_1, ..., w_{\tilde{n}-k-1})_{v_1 \rightarrow 2-type} & w_1 \leq \frac{v_1 + v_2}{2}, \frac{v_1+v_3}{2} \leq w_1' \\
    v_1 + v_3 - 2v_2 - \delta u(v_1,...,v_{k+1}, w_1, ..., w_{\tilde{n}-k-1})_{v_1 \rightarrow 2-type} & w_1 \leq \frac{v_1 + v_2}{2}, w_1' < \frac{v_1+v_3}{2}
    \end{cases}
     \end{split}
     \]
     
We now prove by case-analysis that the condition of \eqref{eq:flippingInequality_u} holds and the above expression is non-negative. 

For the first case, the difference is exactly 0, as ``flipping'' either $v_1$ or $v_2$ does not effect the utilities (we have $\delta u(v_1,...,v_{k+1},w_1,...,w_{\tilde{n}-k-1}) = 2w_1 - 2$).

For the second case, by the condition on $w_1$, $2 + 2v_1 - 4w_1 = 4(\frac{1 + v_1}{2} - w_1) \geq 0$. 

For the remaining cases we first note that whenever $w_1 \leq v_1$, we have
$$
    \delta u(v_1,...,v_k,v_{k+1},w_1,...,w_{\tilde{n}-k-1})_{v_1 \rightarrow 2-type} = 2 - 2v_1 - u_{attack(x,y)} \leq 2 - 2v_1 - (2 - 2v_2) = 2v_2 - 2v_1 \leq 0,$$
    Thus, the third case is at least $2v_1 - 2w_1 \geq 0$. The fourth case is at least $2v_1 - 2w_1' \geq 2v_1 - 2\max \{v_2, w_1\} \geq 0$. The fifth case is at least $v_1 + v_3 - 2v_2 = 2(\frac{v_1 + v_3}{2} - v_2) \geq 2(\frac{v_1 + v_3}{2} - \max \{ v_2, w_1\}) = 2(\frac{v_1 + v_3}{2} - w_1') \geq 0$ by the condition on $w_1'$. 
\end{proof}

\begin{lemma}
\label{lem:monotone_condition}
Recall that $$Q_{\alpha, \beta}^{k,n}(1,1,1) = E^{k,F,n}_{(2,1),truth}[u] - E^{k,F,n}_{(2,1),attack(1,1)}[u],$$
and regard it as a series of real numbers parameterized by k. 
If the series is monotone increasing in k, then if there exists $q$ such that $Q_{\alpha, \beta}^{n}(1,1,1,q) = 0$, it follows that $q^*_{n,g=2,\theta=1,x=1,y=1,F=Uni([0,1])} = q$. 
\end{lemma}
\begin{proof}
It suffices to prove that $\forall q'<q, Q_{\alpha, \beta}^{n}(1,1,1,q') < 0$ and $\forall q' > q, Q_{\alpha, \beta}^{n}(1,1,1,q') > 0$. This is due to first order stochastic dominance of a binomial distribution with a higher $q$ parameter over another binomial distribution with a lower $q$ parameter \cite{Wolfstetter:99}.
\end{proof}

\begin{lemma}
\label{lem:uniform_split_half}
$q^*_{n,g=2,\theta=1,x=1,y=1,F=UNI([0,1])} = \frac{1}{2}$ ($q^*$ is $\frac{1}{2}$ for the split attack under the uniform distribution)
\end{lemma}

   \begin{proof}
   By Lemma~\ref{lem:split_is_monotone_k} and Lemma~\ref{lem:monotone_condition}, it is enough to find such $q$ that $Q_{\alpha, \beta}^{n}(1,1,1,q) = E^{k,F,n}_{(2,1),truth}[u] - E^{k,F,n}_{(2,1),attack(1,1)}[u] = 0$. We show that $q = \frac{1}{2}$ satisfies this equation.
   
   By direct calculation we have
\[
\begin{split}
& F(\tilde{n},k) \stackrel{def}{=} Q_{1,1}^{k,n}(1,1,1) = \begin{cases}
 \frac{1}{(\tilde{n}+1)2^{\tilde{n}-1}} - \frac{2}{\tilde{n}+1} & k = 0 \\
 \frac{8}{\tilde{n}(\tilde{n}+1)} - \frac{2}{\tilde{n}} -  \frac{3}{\tilde{n}(\tilde{n}+1)2^{\tilde{n}-1}} & k = 1 \\
 \frac{1}{\tilde{n}+1} & k = \tilde{n} \\
 \frac{2k}{2^{\tilde{n}-k}}\sum_{i=0}^{\tilde{n}-k}{\binom{\tilde{n}-k}{i}} (\frac{1}{i+k} + \frac{1}{i+k+1}) + \\
 \frac{2k(k-1)(\tilde{n}-k)}{(\tilde{n}+1)(\tilde{n}-k+1)2^{\tilde{n}-k+1}}\sum_{i=0}^{\tilde{n}-k+1}{\binom{\tilde{n}-k+1}{i}} \frac{1}{\tilde{n}-i} - \\
 \frac{2}{\tilde{n}} - \frac{4k(\tilde{n}-k)}{(\tilde{n}-k+1)2^{\tilde{n}-k+1}} \sum_{i=0}^{\tilde{n}-k+1}{\binom{\tilde{n}-k+1}{i}} \frac{1}{i+k} - \\
 \frac{k(k-1)}{(\tilde{n}+1)2^{\tilde{n}-k}}\sum_{i=0}^{\tilde{n}-k}{\binom{\tilde{n}-k}{i}} (\frac{1}{\tilde{n}-i} + \frac{1}{\tilde{n}-i-1}) & o/w. \\
\end{cases}
\end{split}
\]
   
With $q=\frac{1}{2}$, the sum 
$$S[\tilde{n}] \stackrel{def}{=} \sum_{k=2}^{\tilde{n}-1}{\binom{\tilde{n}}{k}}q^k(1-q)^{\tilde{n}-k}F(\tilde{n},k)$$ 
satisfies the recurrence 
\begin{equation}
\label{sum_recurrence}
    -2 (1 + \tilde{n}) S[\tilde{n}] + (2 + \tilde{n}) S[1 + \tilde{n}] = -7 + 3 \cdot 2^(1 - \tilde{n}) + 2 \tilde{n},
    \end{equation}. 
    
We solve the recurrence and check initial
$\tilde{n}$ values. It follows that it holds that 
$$S[\tilde{n}] = -\frac{F(\tilde{n},0) + \tilde{n}F(\tilde{n},1) + F(\tilde{n},\tilde{n})}{2^{\tilde{n}}},$$ 
which means $\forall n\geq 3$,
\[
\begin{split}
& Q_{\alpha, \beta}^{n}(1,1,1,\frac{1}{2}) = \frac{1}{2^{\tilde{n}}} \sum_{k=0}^{\tilde{n}}{\binom{\tilde{n}}{k}}F(\tilde{n},k) = \\
& S[\tilde{n}] + \frac{F(\tilde{n},0) + \tilde{n}F(\tilde{n},1) + F(\tilde{n},\tilde{n})}{2^{\tilde{n}}} = 0.
\end{split}
\]
\end{proof}
We used RISCErgoSum's HolonomicFunctions Mathematica package to derive and solve the recurrence in Equation~\ref{sum_recurrence} \cite{Koutschan09}.

\bibliographystyle{plain}
\bibliography{ref.bib}

\end{document}